\documentclass[11pt,fleqn,a4paper]{article} %,draft

\usepackage{amsmath,amssymb,amsthm,amsfonts} % ,enumerate ,backref ,refcheck
\usepackage[mathscr]{eucal}
\usepackage{graphicx}
\usepackage{caption}
\usepackage{subcaption}
\usepackage[all]{xy}
\usepackage{tikz-cd}

\usepackage{hyperref}
\hypersetup{colorlinks, linkcolor=blue, citecolor=blue, urlcolor=blue}

\allowdisplaybreaks

\makeatletter
\long\def\@makecaption#1#2{%
  \vskip\abovecaptionskip\footnotesize
  \sbox\@tempboxa{#1. #2}%
  \ifdim \wd\@tempboxa >\hsize
    #1. #2\par
  \else
    \global \@minipagefalse
    \hb@xt@\hsize{\hfil\box\@tempboxa\hfil}%
  \fi
  \vskip\belowcaptionskip}
\makeatother

\newcommand{\p}{\partial}
\newcommand{\sgn}{\mathop{\rm sgn}\nolimits}
\newcommand{\ord}{\mathop{\rm ord}\nolimits}

\newcommand{\lsemioplus}{\mathbin{\mbox{$\lefteqn{\hspace{.77ex}\rule{.4pt}{1.2ex}}{\in}$}}}

\newlength{\mylength}
\newcommand{\solution}{\hspace*{-\mylength}\bullet\quad}

\newcommand{\todo}[1][\null]{\ensuremath{\clubsuit}}
\newcommand{\noprint}[1]{}

\newtheorem{theorem}{Theorem}%[section]
\newtheorem{lemma}[theorem]{Lemma}
\newtheorem{corollary}[theorem]{Corollary}

{\theoremstyle{definition}

\newtheorem{remark}[theorem]{Remark}
}

\begin{document}

\par\noindent {\LARGE\bf
Generalized symmetries of remarkable\\ (1+2)-dimensional Fokker--Planck equation
\par}

\vspace{3.5mm}\par\noindent{\large
\large Dmytro R. Popovych$^{\dag\ddag}$, Serhii D. Koval$^{\dag\ddag}$ and Roman O. Popovych$^{\S\ddag}$
}
	
\vspace{4mm}\par\noindent{\it\small
$^\dag$Department of Mathematics and Statistics, Memorial University of Newfoundland,\\
$\phantom{^\dag}$\,St.\ John's (NL) A1C 5S7, Canada
\par}

\vspace{2mm}\par\noindent{\it\small
$^\ddag$Institute of Mathematics of NAS of Ukraine, 3 Tereshchenkivska Str., 01024 Kyiv, Ukraine
\par}

\vspace{2mm}\par\noindent{\it\small
$^\S$\,Mathematical Institute, Silesian University in Opava, Na Rybn\'\i{}\v{c}ku 1, 746 01 Opava, Czech Republic
\par}

\vspace{3mm}\par\noindent
E-mails:
dpopovych@mun.ca, skoval@mun.ca, rop@imath.kiev.ua
	
\vspace{4mm}\par\noindent\hspace*{10mm}\parbox{140mm}{\small
Using an original method,
we find the algebra of generalized symmetries of
a remarkable (1+2)-dimensional ultraparabolic Fokker--Planck equation,
which is also called the Kolmogorov equation and is singled out within
the entire class of ultraparabolic linear second-order partial differential equations with three independent variables
by its wonderful symmetry properties.
It turns out that the essential subalgebra of this algebra,
which consists of linear generalized symmetries,
is generated by the recursion operators associated with the nilradical
of the essential Lie invariance algebra of the Kolmogorov equation,
and the Casimir operator of the Levi factor of the latter algebra unexpectedly arises in the consideration.
We~also establish an isomorphism between this algebra and the Lie algebra associated with the second Weyl algebra,
which provides a dual perspective for studying their properties.
After developing the theoretical background of finding exact solutions
of homogeneous linear systems of differential equations using their linear generalized symmetries,
we efficiently apply it to the Kolmogorov equation.
}\par\vspace{3mm}

\section{Introduction}

Generalized (or higher) symmetries of differential equations
first appeared in the literature in their present form
in Noether's seminal paper~\cite{noet1918} in 1918.
Since then, they have found various applications in symmetry analysis of differential equations, integrability theory,
differential geometry and calculus of variations.
See \cite[pp.~374--379]{olve1993A}
for an excellent exposition on the history and development of the theory
of generalized symmetries and their applications as well as
other monographs on the subject~\cite{blum2010A,blum1989A,boch1999A,ibra1985A,kras1986A}.
At the same time, despite being under study for over a century,
the exhaustive descriptions of generalized symmetry algebras with complete proofs
have only been presented for a small number of specific systems of differential equations.
The main reason for this is the computational complexity inherent in all the problems on finding
objects that are related to systems of differential equations
and defined in the corresponding infinite-order jet spaces.
Notably, the generalized symmetry algebras even of such fundamental and simple models of mathematical physics
as
the linear (1+1)-dimensional heat equation~\cite{kova2023b},
the Burgers equation~\cite{popo2024a},
the linear Korteweg--de Vries equation~\cite{popo2010a} and
the (1+1)-dimensional Klein--Gordon equation~\cite{opan2020e}
were fully described only recently.
See also~\cite{opan2020e} for a review of advances in this field,
\cite{opan2020c} for constructing the generalized symmetry algebra of an isothermal no-slip drift flux model,
and~\cite{east2005a,shap1992a}, \cite{east2008A} and \cite{gove2012a,leva2017a,mich2014a}
for considering generalized symmetries of the Laplace, biharmonic and polyharmonic equations, respectively.

In the present paper, we comprehensively describe the algebra of generalized symmetries of
the Kolmogorov equation~\cite{kolm1934a}
\begin{gather}\label{eq:RemarkableFP}
u_t+xu_y=u_{xx},
\end{gather}
which is an ultraparabolic Fokker--Planck equation.
This equation is singled out within the entire class~$\mathcal U$
of ultraparabolic linear second-order partial differential equations with three independent variables
by its remarkable symmetry properties.
More specifically, it is the unique equation, modulo the point equivalence,
whose essential Lie invariance algebra~$\mathfrak g^{\rm ess}$ is eight-dimensional,
which is the maximum such dimension in the class~$\mathcal U$.
This is why we refer to~\eqref{eq:RemarkableFP} as the {\it remarkable Fokker--Planck equation}.
The above distinguishing properties of the equation~\eqref{eq:RemarkableFP} within the class~$\mathcal U$
are analogous to those of the heat equation
within the class of linear second-order parabolic partial differential equation with two independent variables,
see a discussion in~\cite{kova2023b}.
This is why these two equations are counterparts of each other in their respective classes.
As we will show, this relation also manifests on the level of generalized symmetries,
see Remark~\ref{rem:HeatEqGenSyms}.

The extended classical symmetry analysis of the remarkable Fokker--Planck equation
was carried out in~\cite{kova2023a},
featuring its numerous interesting symmetry properties.
In particular, the point-symmetry pseudogroup~$G$ of~\eqref{eq:RemarkableFP} was computed using the advanced version of the direct method.
One- and two-dimensional subalgebras of the algebra~$\mathfrak g^{\rm ess}$ were classified
modulo the action of the essential subgroup~$G^{\rm ess}$ of~$G$,
followed with the exhaustive classification of Lie reductions of the equation~\eqref{eq:RemarkableFP}
and the construction of wide families of its exact solutions.

The algebra~$\mathfrak g^{\rm ess}$ is wide and has a compound structure.
This provides knowledge of many generalized symmetries of~\eqref{eq:RemarkableFP} for free on the one hand
and complicates the computations and analysis within both the classical and the generalized frameworks on the other hand.
More specifically, the algebra~$\mathfrak g^{\rm ess}$ is isomorphic to a semidirect sum ${\rm sl}(2,\mathbb R)\lsemioplus{\rm h}(2,\mathbb R)$
of the real order-two special linear Lie algebra ${\rm sl}(2,\mathbb R)$
and the real rank-two Heisenberg algebra ${\rm h}(2,\mathbb R)$,
where the action of ${\rm sl}(2,\mathbb R)$ on ${\rm h}(2,\mathbb R)$ is given by the direct sum
of the one- and four-dimensional irreducible representations of ${\rm sl}(2,\mathbb R)$.
Despite the fact that such a structure is similar to those of the essential Lie invariance algebra of the linear (1+1)-dimensional
heat equation, which is isomorphic to ${\rm sl}(2,\mathbb R)\lsemioplus{\rm h}(1,\mathbb R)$,
the corresponding computations are of a higher complexity level.

A preliminary analysis of the generalized symmetry algebra~$\Sigma$
of the remarkable Fokker--Planck equation~\eqref{eq:RemarkableFP}
was carried out in~\cite{kova2023a,kova2024c}.
According to \cite[Proposition~5.22]{olve1993A}, any Lie-symmetry operator%
\footnote{%
A \emph{Lie-symmetry operator} of a homogeneous linear system of differential equations~$\mathcal L$: $\mathrm Lu=0$ is
a first-order linear differential operator~$\mathrm Q$ in total derivatives
such that the tuple of differential functions $\mathrm Qu$ is the characteristic of an (essential) Lie symmetry~of~$\mathcal L$.
}
of the equation~\eqref{eq:RemarkableFP} is its recursion operator.
It was shown in~\cite{kova2024c} that
a complete list of independent operators among such operators
is exhausted by those associated with the canonical basis elements
of the radical~$\mathfrak r$ of~$\mathfrak g^{\rm ess}$.
This is why the associative algebra generated by Lie-symmetry operators of~\eqref{eq:RemarkableFP}
is denoted by~$\Upsilon_{\mathfrak r}$.
We considered the subalgebra~$\Lambda_{\mathfrak r}$ of~$\Sigma$
that consists of the generalized-symmetry vector fields
obtained by the action of the operators from~$\Upsilon_{\mathfrak r}$
on the elementary symmetry vector field~$u\partial_u$ of~\eqref{eq:RemarkableFP}.
We related this subalgebra to generating solutions of the equation~\eqref{eq:RemarkableFP}
via the iterative action by its Lie-symmetry operators.
In this way, taking the group-invariant solutions of the equation~\eqref{eq:RemarkableFP} as seeds,
many more solution families were constructed for it.
Nevertheless, the description of the generalized symmetry algebra
was left in~\cite{kova2023a,kova2024c} as an open problem, which we solve in the present paper.

The algebra $\Sigma$ splits over its infinite-dimensional ideal $\Sigma^{-\infty}$
associated with the linear superposition of the solutions
and constituted by the vector fields $f(t,x,y)\p_u$,
where the parameter function~$f$ runs through the solution set of the equation~\eqref{eq:RemarkableFP}.
Thus, $\Sigma=\Sigma^{\rm ess}\lsemioplus\Sigma^{-\infty}$,
where $\Sigma^{\rm ess}$ is a complementary subalgebra to the ideal $\Sigma^{-\infty}$ in $\Sigma$.
We show that the subalgebra~$\Sigma^{\rm ess}$ coincides with~$\Lambda_{\mathfrak r}$.
The proof of this assertion is surprisingly unusual.
The core of the proof is to show that the entire subalgebra~$\Lambda$ of $\Sigma$
constituted by the linear generalized symmetries of the equation~\eqref{eq:RemarkableFP}
coincides with the algebra~$\Lambda_{\mathfrak r}$.
The latter straightforwardly implies that any subspace consisting of the linear generalized symmetries
of order bounded by a fixed $n\in\mathbb N$ is finite-dimensional,
which allows us to apply the Shapovalov--Shirokov theorem~\cite{shap1992a}
and state that $\Sigma^{\rm ess}=\Lambda$.
Moreover, this approach requires a preliminary study
of the algebra~$\Upsilon_{\mathfrak r}$ using methods from ring theory and algebraic geometry,
which is uncommon for group analysis of differential equations.
The biggest challenge was to analyze how the Casimir operator
of the Levi factor $\mathfrak f\simeq{\rm sl}(2,\mathbb R)$ of~$\mathfrak g^{\rm ess}$
and its multiples are related to the algebra~$\Upsilon_{\mathfrak r}$.
More specifically, the counterpart~$\mathrm C$ of this operator in~$\Upsilon_{\mathfrak r}$
is of degree four as a polynomial, while having order three as a differential operator.
This property impacted constructing a basis and, therefore, computing the dimension
of the subspace~$\Lambda_{\mathfrak r}^n$ of~$\Lambda_{\mathfrak r}$, $n\in\mathbb N_0$,
that is constituted by the elements of~$\Lambda_{\mathfrak r}$ whose order is bounded by~$n$.\looseness=-1

We were impressed to find out that the algebra $\Upsilon_{\mathfrak r}$
is isomorphic to the second Weyl algebra ${\rm W}(2,\mathbb R)$,
which gives rise to the isomorphism between the algebra $\Lambda$
and the Lie algebra ${\rm W}(2,\mathbb R)^{\mbox{\tiny$(-)$}}$ associated with ${\rm W}(2,\mathbb R)$.
Due to this, we straightforwardly obtain a number of properties of the algebra~$\Lambda$.
In particular, it is two-generated and $\mathbb Z$-graded, its center is one-dimensional,
and its quotient algebra by the center is simple.
It was proved in Corollary~21 of the fourth arXiv version of~\cite{kova2023b}
that the algebra of the linear generalized symmetries of the linear (1+1)-dimensional heat equation
is isomorphic to the Lie algebra ${\rm W}(1,\mathbb R)^{\mbox{\tiny$(-)$}}$
associated with the first Weyl algebra ${\rm W}(1,\mathbb R)$,
see Remark~\ref{rem:HeatEqGenSyms} below for more details.
The above facts strengthen the connection between the heat equation and the remarkable Fokker--Planck equation
from the point of view of generalized symmetries,
which we also fit into the wider framework of (ultra)parabolic linear second-order partial differential equations
with an arbitrary number of independent variables in the conclusion.

The paper is organized as follows.
In Section~\ref{sec:RemarkableFPMIA}, we present the maximal Lie invariance algebra
of the remarkable Fokker--Planck equation~\eqref{eq:RemarkableFP} and describe its key properties.
This is followed by the study of the associative algebra $\Upsilon_{\mathfrak r}$
of differential operators generated by the
Lie-symmetry operators associated with the radical $\mathfrak r$ of $\mathfrak g^{\rm ess}$.
We construct, in an explicit form, a basis of the algebra~$\Upsilon_{\mathfrak r}$
and, for each $n\in\mathbb N_0$,
a basis of its subspace of differential operators of order less than or equal to~$n$.
We also show that the algebra~$\Upsilon_{\mathfrak r}$ is isomorphic to the second Weyl algebra ${\rm W}(2,\mathbb R)$.
Section~\ref{sec:RemarkableFPPolynomialSols} is devoted to the study
of the polynomial solutions of~\eqref{eq:RemarkableFP}.
The results of this section are used in Section~\ref{sec:RemarkableFPGenSyms}
in the course of proving the assertion that the algebra~$\Lambda$
coincides with~$\Lambda_{\mathfrak r}$.
The latter straightforwardly leads to the description of the algebra~$\Sigma$
of the generalized symmetries of the equation~\eqref{eq:RemarkableFP}.
Proving that the algebra~$\Lambda$ is isomorphic to the Lie algebra~${\rm W}(2,\mathbb R)^{\mbox{\tiny$(-)$}}$
allows us to transfer the known results about~${\rm W}(2,\mathbb R)^{\mbox{\tiny$(-)$}}$ to~$\Lambda$
and, conversely, to study~${\rm W}(2,\mathbb R)^{\mbox{\tiny$(-)$}}$
from the perspective of its explicit faithful realization~$\Lambda$.
Section~\ref{sec:ExactSolutions} begins with a review of symmetry approaches
for the construction of exact solutions of systems of differential equations.
We also essentially develop the theoretical background of finding
solutions of homogeneous linear systems of differential equations
using their linear generalized symmetries.
Then the developed tools are efficiently applied to the equation~\eqref{eq:RemarkableFP}.
The results of the paper and possible avenues for future research are discussed in Section~\ref{sec:Conclusion}.

\section{Lie-symmetry operators}\label{sec:RemarkableFPMIA}

The maximal Lie invariance algebra of the equation~\eqref{eq:RemarkableFP} is
(see, e.g., \cite{kova2013a})
\begin{gather*}%\label{eq:RemarkableFPMIA}
\mathfrak g:=\langle \mathcal P^t,\,\mathcal D,\,\mathcal K,\,
\mathcal P^3,\,\mathcal P^2,\,\mathcal P^1,\,\mathcal P^0,\,\mathcal I,\mathcal Z(f)\rangle,
\end{gather*}
where
\begin{gather*}
\mathcal P^t =\p_t,\ \
\mathcal D   =2t\p_t+x\p_x+3y\p_y-2u\p_u,\ \
\mathcal K   =t^2\p_t+(tx+3y)\p_x+3ty\p_y-(x^2\!+2t)u\p_u,\\[.5ex]
\mathcal P^3 =3t^2\p_x+t^3\p_y+3(y-tx)u\p_u,\ \
\mathcal P^2 =2t\p_x+t^2\p_y-xu\p_u,\ \
\mathcal P^1 =\p_x+t\p_y,\ \
\mathcal P^0 =\p_y,\\[.5ex]
\mathcal I   =u\p_u,\quad
\mathcal Z(f)=f(t,x,y)\p_u.
\end{gather*}
Here the parameter function $f$ of~$(t,x,y)$ runs through the solution set of the equation~\eqref{eq:RemarkableFP}.
	
The vector fields $\mathcal Z(f)$ constitute the infinite-dimensional abelian ideal $\mathfrak g^{\rm lin}$ of~$\mathfrak g$
associated with the linear superposition of solutions of~\eqref{eq:RemarkableFP}, $\mathfrak g^{\rm lin}:=\{\mathcal Z(f)\}$.
Thus, the algebra $\mathfrak g$ can be represented as a semidirect sum, $\mathfrak g=\mathfrak g^{\rm ess}\lsemioplus\mathfrak g^{\rm lin}$,
where
\begin{gather*}%\label{eq:RemarkableFPEssA}
\mathfrak g^{\rm ess}=\langle\mathcal P^t,\mathcal D,\mathcal K,\mathcal P^3,\mathcal P^2,\mathcal P^1,\mathcal P^0,\mathcal I\rangle
\end{gather*}
is an (eight-dimensional) subalgebra of $\mathfrak g$,
called the essential Lie invariance algebra of~\eqref{eq:RemarkableFP}.
	
Up to the skew-symmetry of the Lie bracket, the nonzero commutation relations between the basis vector fields of $\mathfrak g^{\rm ess}$ are the following:
\begin{gather*}
[\mathcal P^t,\mathcal D]  = 2\mathcal P^t,\quad
[\mathcal P^t,\mathcal K]  =  \mathcal D,\quad
[\mathcal D,  \mathcal K]  = 2\mathcal K,\\[.5ex]
[\mathcal P^t,\mathcal P^3]= 3\mathcal P^2,\quad
[\mathcal P^t,\mathcal P^2]= 2\mathcal P^1,\quad	
[\mathcal P^t,\mathcal P^1]=  \mathcal P^0,\\[.5ex]
[\mathcal D,\mathcal P^3]  = 3\mathcal P^3,\quad		
[\mathcal D,\mathcal P^2]  =  \mathcal P^2,\quad
[\mathcal D,\mathcal P^1]  =- \mathcal P^1,\quad
[\mathcal D,\mathcal P^0]  =-3\mathcal P^0,\\[.5ex]
[\mathcal K,\mathcal P^2]  =- \mathcal P^3,\quad
[\mathcal K,\mathcal P^1]  =-2\mathcal P^2,\quad
[\mathcal K,\mathcal P^0]  =-3\mathcal P^1,\\[.5ex]
[\mathcal P^1,\mathcal P^2]=- \mathcal I,\quad
[\mathcal P^0,\mathcal P^3]= 3\mathcal I.
\end{gather*}
	
The algebra $\mathfrak g^{\rm ess}$ is nonsolvable.
Its Levi decomposition is given by $\mathfrak g^{\rm ess}=\mathfrak f\lsemioplus\mathfrak r$,
where the radical~$\mathfrak r$ of~$\mathfrak g^{\rm ess}$ coincides with the nilradical of~$\mathfrak g^{\rm ess}$ and
is spanned by the vector fields $\mathcal P^3$, $\mathcal P^2$, $\mathcal P^1$, $\mathcal P^0$ and~$\mathcal I$,
\[
\mathfrak r=\langle\mathcal P^3,\mathcal P^2,\mathcal P^1,\mathcal P^0,\mathcal I\rangle.
\]
The Levi factor $\mathfrak f=\langle\mathcal P^t,\mathcal D,\mathcal K\rangle$ of~$\mathfrak g^{\rm ess}$
is isomorphic to ${\rm sl}(2,\mathbb R)$,
the radical~$\mathfrak r$ of~$\mathfrak g^{\rm ess}$ is isomorphic to the rank-two Heisenberg algebra ${\rm h}(2,\mathbb R)$,
and the real representation of the Levi factor~$\mathfrak f$ on the radical~$\mathfrak r$
coincides, in the basis $(\mathcal P^3,\mathcal P^2,\mathcal P^1,\mathcal P^0,\mathcal I)$,
with the real representation $\rho_3\oplus \rho_0$ of~${\rm sl}(2,\mathbb R)$.
Here $\rho_n$ is the standard real irreducible representation of~${\rm sl}(2,\mathbb R)$ in the $(n+1)$-dimensional vector space.
More specifically,
\[
\rho_n( \mathcal P^t)_{ij}=(n-j)\delta_{i,j+1},\quad
\rho_n( \mathcal D)_{ij}  =(n-2j)\delta_{ij},\quad
\rho_n(-\mathcal K)_{ij}  =j\delta_{i+1,j},
\]
where $i,j\in\{1,2,\dots,n+1\}$, $n\in\mathbb N_0:=\mathbb N\cup\{0\}$,
and $\delta_{kl}$ is the Kronecker delta, i.e., $\delta_{kl}=1$ if $k=l$ and $\delta_{kl}=0$ otherwise, $k,l\in\mathbb N_0$.
Thus, the entire algebra~$\mathfrak g^{\rm ess}$ is isomorphic to the algebra
$L_{8,19}\simeq{\rm sl}(2,\mathbb R)\lsemioplus_{\rho_2\oplus \rho_0}{\rm h}(1,\mathbb R)$
from the classification of indecomposable Lie algebras of dimensions up to eight
with nontrivial Levi decompositions, which was carried out in~\cite{turk1988a}.

Lie algebras whose Levi factors are isomorphic to the algebra ${\rm sl}(2,\mathbb R)$ often arise
within the field of group analysis of differential equations
as Lie invariance algebras of parabolic partial differential equations.
At the same time, the action of Levi factors on the corresponding radicals is usually described
in terms of the representations $\rho_0$, $\rho_1$, $\rho_2$ or their direct sums,
cf.\ the essential Lie invariance algebra of the linear (1+1)-dimensional heat equation,
which is isomorphic to the so-called special Galilei algebra
${\rm sl}(2,\mathbb R)\lsemioplus_{\rho_1\oplus \rho_0}{\rm h}(1,\mathbb R)$
\cite[Section~3]{kova2023b}.
To the best of our knowledge, algebras analogous to~$\mathfrak g^{\rm ess}$ had not been studied
in group analysis from the point of view of their subalgebra structure before~\cite{kova2023a}.
See the conclusion for the discussion on (ultra)parabolic second-order multidimensional linear
partial differential equations with essential Lie invariance algebras
isomorphic to ${\rm sl}(2,\mathbb R)\lsemioplus_{\rho_{2n-1}\oplus \rho_0}{\rm h}(n,\mathbb R)$,
where ${\rm h}(n,\mathbb R)$ denotes the rank-$n$ Heisenberg algebra, $n\in\mathbb N$.

Consider the Lie-symmetry operators of~\eqref{eq:RemarkableFP}
that are associated with the Lie-symmetry vector fields~$-\mathcal P^3$, $-\mathcal P^2$, $-\mathcal P^1$, $-\mathcal P^0$
and~$-\mathcal P^t$, $-\mathcal D$, $-\mathcal K$
(here we take minuses for a nicer representation of differential operators),
\begin{gather*}
\mathrm P^3:=3t^2\mathrm D_x+t^3\mathrm D_y-3(y-tx),\ \
\mathrm P^2:=2t\mathrm D_x+t^2\mathrm D_y+x,\ \
\mathrm P^1:=\mathrm D_x+t\mathrm D_y,\ \
\mathrm P^0:=\mathrm D_y,\\
\mathrm P^t:=\mathrm D_t,\ \
\mathrm D  :=2t\mathrm D_t+x\mathrm D_x+3y\mathrm D_y+2,\ \
\mathrm K  :=t^2\mathrm D_t+(tx+3y)\mathrm D_x+3ty\mathrm D_y+x^2\!+2t.
\end{gather*}

The associative operator algebra $\Upsilon_{\mathfrak r}$
generated by the operators~$\mathrm P^3$, $\mathrm P^2$, $\mathrm P^1$ and $\mathrm P^0$
admits the following presentation:
\begin{gather}\label{eq:UpsilonPresentation}
\begin{split}
\Upsilon_{\mathfrak r}=\big\langle
&\mathrm P^3,\mathrm P^2,\mathrm P^1,\mathrm P^0\mid
\\
&[\mathrm P^3,\mathrm P^0]=3,\,
[\mathrm P^1,\mathrm P^2]=1,\,
%[\mathrm P^3,\mathrm P^2]=0,\,[\mathrm P^3,\mathrm P^1]=0,\,[\mathrm P^2,\mathrm P^0]=0,\,[\mathrm P^1,\mathrm P^0]=0
[\mathrm P^3,\mathrm P^2]=[\mathrm P^3,\mathrm P^1]=[\mathrm P^2,\mathrm P^0]=[\mathrm P^1,\mathrm P^0]=0
\big\rangle.
\end{split}
\end{gather}
We begin describing the properties of the algebra $\Upsilon_{\mathfrak r}$ with finding its explicit basis.

\begin{lemma}\label{lem:BasisUpsilonR}
Fixed any ordering $(\mathrm Q^0,\mathrm Q^1,\mathrm Q^2,\mathrm Q^3)$ of $\{\mathrm P^0,\mathrm P^1,\mathrm P^2,\mathrm P^3\}$,
$\mathrm Q^0<\mathrm Q^1<\mathrm Q^2<\mathrm Q^3$,
the monomials $\mathbf Q^\alpha:=(\mathrm Q^0)^{\alpha_0}(\mathrm Q^1)^{\alpha_1}(\mathrm Q^2)^{\alpha_2}(\mathrm Q^3)^{\alpha_3}$
with $\alpha=(\alpha_0,\alpha_1,\alpha_2,\alpha_3)\in\mathbb N_0^{\,\,4}$
constitute a basis of the algebra~$\Upsilon_{\mathfrak r}$.
\end{lemma}

\begin{proof}
The required claim follows from Bergman's diamond lemma \cite[Theorem~1.2]{berg1978a},
see therein for all the related notions.
Indeed, under the interpretation of the algebra~$\Upsilon_{\mathfrak r}$ following \cite[Section~3]{berg1978a},
there are exactly four overlap ambiguities, which are related to the products
$\mathrm Q^3\mathrm Q^2\mathrm Q^1$,
$\mathrm Q^3\mathrm Q^2\mathrm Q^0$,
$\mathrm Q^3\mathrm Q^1\mathrm Q^0$ and
$\mathrm Q^2\mathrm Q^1\mathrm Q^0$,
and each of them is resolvable.
\end{proof}

By default, we use the ordering $\mathrm P^3<\mathrm P^2<\mathrm P^1<\mathrm P^0$.

\begin{lemma}\label{lem:IsoToUnivEnv}
In the sense of unital algebras, the algebra $\Upsilon_{\mathfrak r}$ is isomorphic
to the quotient algebra of the universal enveloping algebra $\mathfrak U(\mathfrak r)$ of~$\mathfrak r$
by the two-sided ideal $(\iota(\mathcal I)+1)$ generated by $\iota(\mathcal I)+1$,
$\Upsilon_{\mathfrak r}\simeq\mathfrak U(\mathfrak r)/(\iota(\mathcal I)+1)$,
where $\iota\colon\mathfrak r\hookrightarrow\mathfrak U(\mathfrak r)$
is the canonical embedding of the Lie algebra $\mathfrak r$ in its universal enveloping algebra~$\mathfrak U(\mathfrak r)$.
Moreover, this defines an isomorphism between the associated Lie algebras $\Upsilon_{\mathfrak r}^{\mbox{\tiny$(-)$}}$
and $\big(\mathfrak U(\mathfrak r)/(\iota(\mathcal I)+1)\big)^{\mbox{\tiny$(-)$}}$.
\end{lemma}

\begin{proof}
The correspondence $\mathcal P^j\mapsto\mathrm P^j$, $j=0,1,2,3$, and $\mathcal I\mapsto-1$
linearly extends to the Lie algebra homomorphism~$\varphi$ from $\mathfrak r$
to the Lie algebra $\Upsilon_{\mathfrak r}^{\mbox{\tiny$(-)$}}$
associated with the associative algebra $\Upsilon_{\mathfrak r}$.
By the universal property of the universal enveloping algebra~$\mathfrak U(\mathfrak r)$,
the Lie algebra homomorphism~$\varphi$ extends to the (unital) associative algebra homomorphism
$\hat\varphi\colon\mathfrak U(\mathfrak r)\to\Upsilon_{\mathfrak r}$,
i.e., $\varphi=\hat\varphi\circ\iota$ as homomorphisms of vector spaces.
Since the algebra $\Upsilon_{\mathfrak r}$ is generated by $\varphi(\mathfrak r)$,
the homomorphism $\hat\varphi$ is surjective.

For the rest of the proof, we identify $\mathfrak r$ with its image under the map $\iota$ in $\mathfrak U(\mathfrak r)$.
It is clear that $(\mathcal I+1)\subset\ker\hat\varphi$.
To show the reverse inclusion,
consider an arbitrary polynomial $Q\in\mathfrak U(\mathfrak r)$,
which in view of the Poincar\'e--Birkhoff--Witt theorem takes the form
\[
Q=c_{i_3i_2i_1i_0j}(\mathcal P^3)^{i_3}(\mathcal P^2)^{i_2}(\mathcal P^1)^{i_1}(\mathcal P^0)^{i_0}\mathcal I^j
\]
with a  finite number of nonzero coefficients $c_{i_3i_2i_1i_0j}$, and assume that $Q\in\ker\hat\varphi$,
\begin{gather*}
\hat\varphi(Q)=(-1)^jc_{i_3i_2i_1i_0j}(\mathrm P^3)^{i_3}(\mathrm P^2)^{i_2}(\mathrm P^1)^{i_1}(\mathrm P^0)^{i_0}=0.
\end{gather*}
Here and in what follows we assume summation with respect to repeated indices.
In view of Lemma~\ref{lem:BasisUpsilonR},
we have $(-1)^jc_{i_3i_2i_1i_0j}=0$ for each fixed tuple $(i_3,i_2,i_1,i_0)$.
Therefore,
\begin{gather*}
\begin{split}
Q&{}=c_{i_3i_2i_1i_0j}(\mathcal P^3)^{i_3}(\mathcal P^2)^{i_2}(\mathcal P^1)^{i_1}(\mathcal P^0)^{i_0}\mathcal I^j
-(-1)^jc_{i_3i_2i_1i_0j}(\mathcal P^3)^{i_3}(\mathcal P^2)^{i_2}(\mathcal P^1)^{i_1}(\mathcal P^0)^{i_0}
\\&{}=c_{i_3i_2i_1i_0j}(\mathcal P^3)^{i_3}(\mathcal P^2)^{i_2}(\mathcal P^1)^{i_1}(\mathcal P^0)^{i_0}(\mathcal I^j-(-1)^j).
\end{split}
\end{gather*}
For each~$j$, the factor $\mathcal I^j-(-1)^j$ is divisible by $\mathcal I+1$.
Therefore, $\ker\hat\varphi=(\mathcal I+1)$
and the isomorphism $\Upsilon_{\mathfrak r}\simeq\mathfrak U(\mathfrak r)/(\mathcal I+1)$
follows from the first isomorphism theorem for associative algebras.

The isomorphism between the associated Lie algebras
$\Upsilon_{\mathfrak r}^{\mbox{\tiny$(-)$}}$ and $\big(\mathfrak U(\mathfrak r)/(\mathcal I+1)\big)^{\mbox{\tiny$(-)$}}$
follows from the fact that, by definition,
the Lie brackets on these algebras are the ring-theoretic commutators
on the corresponding associative algebras.
\end{proof}

\begin{remark}\label{rem:WeylAlgebra}
Recall the definition of the $n$th Weyl algebra ${\rm W}(n,\mathbb R)$.
It is the quotient of the free associative $\mathbb R$-algebra on the alphabet
$\{\hat p_1,\dots,\hat p_n,\hat q_1,\dots,\hat q_n\}$
by the two-side ideal generated by
$\hat p_i\hat p_j-\hat p_j\hat p_i$, $\hat q_i\hat q_j-\hat q_j\hat q_i$
and $\hat p_i\hat q_j-\hat q_j\hat p_i-\delta_{ij}$.
Here and in the rest of this remark, the indices~$i$ and~$j$ run from 1 to~$n$.
Recall that $\delta_{ij}$ denotes the Kronecker delta.
Hence the algebra ${\rm W}(n,\mathbb R)$ admits the presentation
\[
{\rm W}(n,\mathbb R)=
\langle\hat p_1,\dots,\hat p_n,\hat q_1,\dots,\hat q_n\mid
\hat p_i\hat p_j-\hat p_j\hat p_i=\hat q_i\hat q_j-\hat q_j\hat q_i=0,\,\hat p_i\hat q_j-\hat q_j\hat p_i=\delta_{ij}
\rangle.
\]
This algebra can be related to the quotient of the universal enveloping algebra
of the rank-$n$ Heisenberg Lie algebra ${\rm h}(n,\mathbb R)$.
More specifically, let the elements~$p_i$, $q_i$ and $c$ constitute
the canonical basis of the Lie algebra ${\rm h}(n,\mathbb R)$,
and thus they satisfy the commutation relations $[p_i,p_j]=[q_i,q_j]=0$ and $[p_i,q_j]=\delta_{ij}c$.
The $n$th Weyl algebra ${\rm W}(n,\mathbb R)$ is the quotient
of the universal enveloping algebra $\mathfrak U\big({\rm h}(n,\mathbb R)\big)$ of ${\rm h}(n,\mathbb R)$
by the two-sided ideal $(c-1)$ generated by $c-1$,
\smash{${\rm W}(n,\mathbb R):=\mathfrak U\big({\rm h}(n,\mathbb R)\big)/(c-1)$}.
The canonical basis in ${\rm W}(n,\mathbb R)$ consists of monomials $q^\kappa p^\lambda$,
where
$q^\kappa:=q_1^{\kappa_1}\cdots q_n^{\kappa_n}$,
$p^\lambda:=p_1^{\lambda_1}\cdots p_n^{\lambda_n}$
and
$\kappa:=(\kappa_1,\dots,\kappa_n)$, $\lambda=(\lambda_1,\dots,\lambda_n)$
are multiindices running through $\mathbb N_0^n$.
\end{remark}

The commutation relations in the above basis of ${\rm W}(n,\mathbb R)$ take the form
\begin{gather}\label{eq:CommutatorsWn-}
[q^\kappa p^\lambda,q^{\kappa'}p^{\lambda'}]%\\
=\sum_{\nu\in\mathbb N_0^n} \nu!\left(\binom{\kappa'}\nu\binom\lambda\nu-\binom\kappa\nu\binom{\lambda'}\nu\right)
q^{\kappa+\kappa'-\nu}p^{\lambda+\lambda'-\nu}.
\end{gather}
Note that the series in the commutator $[q^\kappa p^\lambda,q^{\kappa'}p^{\lambda'}]$
is in fact the finite sum with the multiindex~$\nu$ satisfying the inequalities
$\nu\leqslant\kappa$ and $\nu\leqslant\lambda'$ or $\nu\leqslant\kappa'$ and $\nu\leqslant\lambda$
since $\binom\kappa\nu=0$ if $\nu$ does not satisfy the inequality $\nu\leqslant\kappa$.

\begin{remark}\label{rem:OpWeylAlgebra}
The opposite algebra ${\rm W}(n,\mathbb R)^{\rm op}$ of ${\rm W}(n,\mathbb R)$ admits the presentation
\[
{\rm W}(n,\mathbb R)^{\rm op}=
\langle \check p_1,\dots,\check p_n,\check q_1,\dots,\check q_n\mid
\check p_i\check p_j-\check p_j\check p_i=\check q_i\check q_j-\check q_j\check q_i=0,\,
\check p_i\check q_j-\check q_j\check p_i=-\delta_{ij}
\rangle.
\]
This results in the isomorphism ${\rm W}(n,\mathbb R)^{\rm op}\simeq\mathfrak U\big({\rm h}(n,\mathbb R)\big)/(c+1)$
defined on the algebra generators by the correspondence $\check p_i\mapsto p_i$, $\check q_i\mapsto q_i$.
It is clear that the algebras ${\rm W}(n,\mathbb R)$ and ${\rm W}(n,\mathbb R)^{\rm op}$ are isomorphic, where
the simplest isomorphism is given by permuting the $p$- and $q$-tuples,
$\hat p_i\leftrightarrow\check q_i$,~$\hat q_i\leftrightarrow\check p_i$.
\end{remark}

\begin{corollary}\label{cor:IsoToWeyl}
The algebra~$\Upsilon_{\mathfrak r}$ is isomorphic to the opposite of the second Weyl algebra
and hence to the second Weyl algebra itself,
$\Upsilon_{\mathfrak r}\simeq{\rm W}(2,\mathbb R)^{\rm op}\simeq{\rm W}(2,\mathbb R)$.
\end{corollary}

The explicit isomorphisms in Corollary~\ref{cor:IsoToWeyl} are established,
e.g., by the correspondence on the level of generators in the following way:
\[%\label{eq:UpsilonIsoToWeyl}
(\tfrac13\mathrm P^3,\mathrm P^2,\mathrm P^1,\mathrm P^0)
\mapsto(\check q_1,\check p_2,\check q_2,\check p_1)
\mapsto(\hat p_1,\hat q_2,\hat p_2,\hat q_1).
\]

The algebra~$\Upsilon_{\mathfrak r}$ possesses two natural filtrations,
\begin{gather*}
F_1\colon\quad
\Upsilon_{\mathfrak r}=\bigcup_{n\in\mathbb N_0}\Upsilon^{\ord}_n,\quad
\Upsilon^{\ord}_n:=\{\mathrm Q\in\Upsilon_{\mathfrak r}\mid\mathop{\rm ord}\mathrm Q\leqslant n\},
\\
F_2\colon\quad
\Upsilon_{\mathfrak r}=\bigcup_{n\in\mathbb N_0}\Upsilon^{\deg}_n,\quad
\Upsilon^{\deg}_n:=\{\mathrm Q\in\Upsilon_{\mathfrak r}\mid\mathop{\rm deg}\mathrm Q\leqslant n\},
\end{gather*}
where $\mathop{\rm ord}\mathrm Q$ is the order of~$\mathrm Q$ as a differential operator
and $\mathop{\rm deg}\mathrm Q$ is the degree of~$\mathrm Q$ as a (noncommutative) polynomial
in $\{\mathrm P^0,\mathrm P^1,\mathrm P^2,\mathrm P^3\}$.
It is clear that $\mathop{\rm ord}\mathrm Q\leqslant\mathop{\rm deg}\mathrm Q$ for any $\mathrm Q\in\Upsilon_{\mathfrak r}$.
Therefore, for each $n\in\mathbb N_0$ we have the inclusion $\Upsilon^{\deg}_n\subseteq\Upsilon^{\ord}_n$.
The (unordered) basis of the space~$\Upsilon^{\deg}_n$
that corresponds to the ordering $\mathrm P^3<\mathrm P^2<\mathrm P^1<\mathrm P^0$
is the set
\[\{(\mathrm P^3)^{i_3}(\mathrm P^2)^{i_2}(\mathrm P^1)^{i_1}(\mathrm P^0)^{i_0}\mid i_3+i_2+i_1+i_0\leqslant n\},\]
which is the restriction of the corresponding basis of the algebra~$\Upsilon_{\mathfrak r}$
to the subspace~$\Upsilon^{\deg}_n$.

The description of bases of the subspaces $\Upsilon^{\ord}_n$, $n\in\mathbb N_0$,
is more complicated.
To construct such bases, we should consider a distinguish element~$\mathrm C$ of~$\Upsilon_{\mathfrak r}$.
On solutions of the equation~\eqref{eq:RemarkableFP},
its Lie-symmetry operators~$\mathrm P^t$, $\mathrm D$ and~$\mathrm K$
associated with its Lie symmetries~$-\mathcal P^t$, $-\mathcal D$ and $-\mathcal K$
are equivalent to the elements
\begin{gather}\label{eq:PDKRepresentations}
\begin{split}&
\hat{\mathrm P}^t:=(\mathrm P^1)^2-\mathrm P^2\mathrm P^0=\mathrm D_x^2-x\mathrm D_y,
\\&
\hat{\mathrm D}  :=\mathrm P^2\mathrm P^1-\mathrm P^3\mathrm P^0+2=2t\mathrm D_x^2+x\mathrm D_x+(3y-2tx)\mathrm D_y+2,
\\&
\hat{\mathrm K}  :=(\mathrm P^2)^2-\mathrm P^3\mathrm P^1=t^2\mathrm D_x^2+(3y+tx)\mathrm D_x+t(3y-tx)\mathrm D_y+x^2+2t
%\qquad \mathrm P^0\mathrm P^3=\mathrm P^3\mathrm P^0-3,\quad \mathrm P^1\mathrm P^2=\mathrm P^2\mathrm P^1+1,\\
\end{split}
\end{gather}
of the associative algebra~$\Upsilon_{\mathfrak r}$, respectively.
The associative algebra~$\Upsilon_{\mathfrak f}$ generated by~$\hat{\mathrm P}^t$, $\hat{\mathrm D}$ and~$\hat{\mathrm K}$
is isomorphic to the universal enveloping algebra~$\mathfrak U(\mathfrak f)$ of the Levi factor~$\mathfrak f$.
In other words, the algebra~$\Upsilon_{\mathfrak r}$ contains
an isomorphic copy~$\Upsilon_{\mathfrak f}$ of the universal enveloping algebra~$\mathfrak U(\mathfrak f)$.
This allows us to consider the counterpart
of the Casimir operator $\mathrm D^2-2(\mathrm K\mathrm P^t+\mathrm P^t\mathrm K)$
of the Levi factor~$\mathfrak f$ inside the algebra~$\Upsilon_{\mathfrak r}$.
This operator
is equivalent on the solutions of~\eqref{eq:RemarkableFP} to the operator
\begin{gather*}
\begin{split}
\mathrm C&:=\hat{\mathrm D}^2-2(\hat{\mathrm K}\hat{\mathrm P}^t+\hat{\mathrm P}^t\hat{\mathrm K})
\noprint{
\\&\phantom{:}
=(\mathrm P^2\mathrm P^1-\mathrm P^3\mathrm P^0+2)^2
-2((\mathrm P^2)^2-\mathrm P^3\mathrm P^1)((\mathrm P^1)^2-\mathrm P^2\mathrm P^0)
-2((\mathrm P^1)^2-\mathrm P^2\mathrm P^0)((\mathrm P^2)^2-\mathrm P^3\mathrm P^1)
\\&\phantom{:} %%%%%
=(\mathrm P^2)^2(\mathrm P^1)^2+\mathrm P^2\mathrm P^1-2\mathrm P^3\mathrm P^2\mathrm P^1\mathrm P^0
+(\mathrm P^3)^2(\mathrm P^0)^2-3\mathrm P^3\mathrm P^0+4\mathrm P^2\mathrm P^1-4\mathrm P^3\mathrm P^0+4
\\&\phantom{:=}
-2(\mathrm P^2)^2(\mathrm P^1)^2+2(\mathrm P^2)^3\mathrm P^0
+2\mathrm P^3(\mathrm P^1)^3-2\mathrm P^3\mathrm P^2\mathrm P^1\mathrm P^0-2\mathrm P^3\mathrm P^0
\\&\phantom{:=}
-2(\mathrm P^2)^2(\mathrm P^1)^2-8\mathrm P^2\mathrm P^1-4
+2\mathrm P^3(\mathrm P^1)^3+2(\mathrm P^2)^3\mathrm P^0
-2\mathrm P^3\mathrm P^2\mathrm P^1\mathrm P^0+6\mathrm P^2\mathrm P^1
}
\\&\phantom{:} %%%%%
=(\mathrm P^3)^2(\mathrm P^0)^2-6\mathrm P^3\mathrm P^2\mathrm P^1\mathrm P^0-3(\mathrm P^2)^2(\mathrm P^1)^2
+4(\mathrm P^2)^3\mathrm P^0+4\mathrm P^3(\mathrm P^1)^3
+3\mathrm P^2\mathrm P^1-9\mathrm P^3\mathrm P^0
\\&\phantom{:} %%%%%
=-12y\mathrm D_x^3-3x^2\mathrm D_x^2+18xy\mathrm D_x\mathrm D_y
+9y^2\mathrm D_y^2+3x\mathrm D_x+(4x^3+27y)\mathrm D_y.
\end{split}
\end{gather*}
We observe an interesting phenomenon in the algebra~$\Upsilon_{\mathfrak r}$.
The element~$\mathrm C$ of~$\Upsilon_{\mathfrak r}$ is a third-order differential operator.
At the same time, it is a linear combination of monomials in $(\mathrm P^3,\mathrm P^2,\mathrm P^1,\mathrm P^0)$ up to degree four and, in view of Lemma~\ref{lem:BasisUpsilonR},
it cannot be represented as a linear combination of monomials of degrees less than or equal to three.
Moreover, it can be proved that modulo linearly recombining with later monomials,
it is a unique element with such property
within the subspace of third-order differential operators in~$\Upsilon_{\mathfrak r}$.%
\footnote{%
We first derived this claim after computing the space~$\Sigma^3$ of generalized symmetries
of the equation~\eqref{eq:RemarkableFP},
see the notation in Section~\ref{sec:RemarkableFPGenSyms} below.
It is also an obvious consequence of Theorem~\ref{thm:BasisOforderNOperators}.
}
The operator~$\mathrm C$ has a number of other specific properties.
In particular, the only third-order differentiation in it is $\mathrm D_x^3$,
it contains no zero-order term and its coefficients do not depend on~$t$.

\begin{theorem}\label{thm:BasisOforderNOperators}
A basis of the subspace $\Upsilon^{\ord}_n$ of differential operators of order less than or equal to $n\in\mathbb N_0$
in~$\Upsilon_{\mathfrak r}$ is constituted by the products
$(\mathrm P^3)^{i_3}(\mathrm P^2)^{i_2}(\mathrm P^1)^{i_1}(\mathrm P^0)^{i_0}$,
where ${i_0,i_1,i_2,i_3\in\mathbb N_0}$ with $i_0+i_1+i_2+i_3\leqslant n$, and
$\mathrm C^m(\mathrm P^3)^{i_3}(\mathrm P^2)^{i_2}(\mathrm P^1)^{i_1}(\mathrm P^0)^{i_0}$,
where ${i_0,i_1,i_2,i_3\in\mathbb N_0}$ and $m\in\mathbb N$ with $i_0+i_1+i_2+i_3+3m=n$.
\end{theorem}

\begin{proof}
Consider the associated graded algebras~${\rm gr}_1\Upsilon_{\mathfrak r}$ and~${\rm gr}_2\Upsilon_{\mathfrak r}$
of the algebra~$\Upsilon_{\mathfrak r}$ with respect to the filtrations $F_1$ and $F_2$, respectively,
\begin{gather*}
{\rm gr}_1\Upsilon_{\mathfrak r}:=\bigoplus_{n=0}^\infty \Upsilon^{\ord}_n/\Upsilon^{\ord}_{n-1}
\quad\text{and}\quad
{\rm gr}_2\Upsilon_{\mathfrak r}:=\bigoplus_{n=0}^\infty \Upsilon^{\deg}_n/\Upsilon^{\deg}_{n-1},
\end{gather*}
assuming $\Upsilon^{\ord}_{-1}=\Upsilon^{\deg}_{-1}:=\{0\}$.
The algebra~$\Upsilon_{\mathfrak r}$ is related to~${\rm gr}_1\Upsilon_{\mathfrak r}$
and~${\rm gr}_2\Upsilon_{\mathfrak r}$
via the corresponding initial form maps $\psi_i\colon\Upsilon_{\mathfrak r}\to{\rm gr}_i\Upsilon_{\mathfrak r}$,
\[
\psi_1(\mathrm Q):=\pi^{1}_{\ord \mathrm Q-1}(\mathrm Q)\quad
\text{and}\quad
\psi_2(\mathrm Q):=\pi^{2}_{\deg \mathrm Q-1}(\mathrm Q),
\quad
\mathrm Q\in\Upsilon_{\mathfrak r},
\]
where $\pi^{1}_n\colon\Upsilon_{\mathfrak r}\to\Upsilon_{\mathfrak r}/\Upsilon^{\ord}_n$
and $\pi^{2}_n\colon\Upsilon_{\mathfrak r}\to\Upsilon_{\mathfrak r}/\Upsilon^{\deg}_n$
are the canonical projections.
Properties of the commutator of differential operators
and the presentation~\eqref{eq:UpsilonPresentation} of the algebra~$\Upsilon_{\mathfrak r}$ straightforwardly imply
that the algebras~${\rm gr}_1\Upsilon_{\mathfrak r}$ and~${\rm gr}_2\Upsilon_{\mathfrak r}$ are commutative.
Moreover, the algebra ${\rm gr}_2\Upsilon_{\mathfrak r}$ is the polynomial algebra
$\mathbb R[x_0,x_1,x_2,x_3]$ in the variables $x_j:=\psi_2(\mathrm P^j)$, $j=0,1,2,3$.
Extending $\psi_1$ to the algebra of differential operators in the total derivatives
with respect to~$x$ and~$y$ with coefficients depending on $(t,x,y)$,
we denote $X:=\psi_1({\rm D}_x)$ and $Y:=\psi_1({\rm D}_y)$.
Then
\begin{gather*}%\label{eq:ZeroLocusCasimir}
\psi_1(\mathrm P^0):=Y,\quad
\psi_1(\mathrm P^1):=X+tY,\quad
\psi_1(\mathrm P^2):=2tX+t^2Y,\quad
\psi_1(\mathrm P^3):=3t^2X+t^3Y,
\end{gather*}
and the algebra ${\rm gr}_1\Upsilon_{\mathfrak r}$ can be identified with the polynomial algebra
\[\mathbb R[Y,\,X+tY,\,2tX+t^2Y,\,3t^2X+t^3Y].\]

The subspace inclusions $i_n\colon\Upsilon^{\deg}_n\hookrightarrow\Upsilon^{\ord}_n$, $n\in\mathbb N_0\cup\{-1\}$,
jointly give rise to an algebra homomorphism $f\colon{\rm gr}_2\Upsilon_{\mathfrak r}\to{\rm gr}_1\Upsilon_{\mathfrak r}$
that makes the following diagram commutative for each $n\in\mathbb N_0\cup\{-1\}$:
%\\\todo Diagram\todo\noprint{
\begin{center}
\begin{tikzcd}
\Upsilon^{\deg}_n\arrow[d,"\psi_2"]\arrow[r,hookrightarrow,"i_n"]&\Upsilon^{\ord}_n\arrow[d,"\psi_1"]\\
{\rm gr}_2\Upsilon_{\mathfrak r}\arrow[r,"f"]  &{\rm gr}_1\Upsilon_{\mathfrak r}
\end{tikzcd}
\end{center}
%}
The map $f$ is defined elementwise via the correspondence
\[\mathrm Q+\Upsilon^{\deg}_{\deg\mathrm Q-1}\mapsto\psi_1(\mathrm Q)+\Upsilon^{\ord}_{\deg\mathrm Q-1}.\]
It is straightforward to verify that it is a well-defined unital homomorphism of associative algebras,
and $f(x_j)=\psi_1(\mathrm P^j)$, $j=0,1,2,3$.
In other words, the image of a differential operator~$\mathrm Q\in\Upsilon_{\mathfrak r}$ under the composition $f\circ\psi_2$
is its formal symbol if $\ord \mathrm Q=\deg \mathrm Q$, and it is zero otherwise.

The property $f\circ\psi_2(\mathrm C)=0$ of the Casimir element $\mathrm C\in\Upsilon_{\mathfrak r}$
is equivalent to the fact that the solution set of the polynomial equation $\check{\mathrm C}=0$, where
\[
\check{\mathrm C}:=\psi_2(\mathrm C)=x_3^2x_0^2-6x_3x_2x_1x_0-3x_2^2x_1^2+4x_2^3x_0+4x_3x_1^3,
\]
is a hypersurface in~$\mathbb R^4$ with the parameterization
\begin{gather*}
x_3=3t^2X+t^3Y,\quad
x_2=2tX+t^2Y,\quad
x_1=X+tY,\quad
x_0=Y,
\end{gather*}
where $(t,X,Y)$ is considered as the coordinate tuple of the affine space $\mathbb R^3$.

If $\deg \mathrm Q>\ord\mathrm Q$, then $f\circ\psi_2(\mathrm Q)=0$,
and thus the zero locus of the polynomial $\check{\mathrm C}$
is contained in the zero locus of the polynomial $\check{\mathrm Q}:=\psi_2(\mathrm Q)$.
In other words, the vanishing ideal of the hypersurface~$\check{\mathrm Q}=0$ in the polynomial algebra $\mathbb R[x_0,x_1,x_2,x_3]$
is contained in the vanishing ideal of the hypersurface~$\check{\mathrm C}=0$ in this algebra.
Therefore, by Hilbert's Nullstellensatz in the form~\cite[Chapter~VII, Theorem~14]{zari1960A},
the polynomial $\check{\mathrm Q}$ belongs to the radical of the principal ideal $I:=(\check{\mathrm C})$ in $\mathbb R[x_0,x_1,x_2,x_3]$,
i.e., there exists $m\in\mathbb N$ such that $\check{\mathrm Q}^m\in I$.

We show that the polynomial~$\check{\mathrm C}$ is irreducible.
Assume to the contrary that it is reducible.
Since the multipliers $x_3^2$ and $x_0^2$ appear only in the monomial $x_3^2x_0^2$ in~$\check{\mathrm C}$
and both~$x_0$ and~$x_3$ do not divide~$\check{\mathrm C}$,
the only possible factorization of~$\check{\mathrm C}$ is
\[(x_3x_0+p)(x_3x_0+q)\]
for some homogeneous second-degree polynomials $p,q\in\mathbb R[x_0,x_1,x_2,x_3]$
that are affine with respect to $(x_0,x_3)$.
Hence $p+q=-6x_1x_2$ and $pq=-3x_2^2x_1^2+4x_2^3x_0+4x_3x_1^3$.
Up to the permutation of~$p$ and~$q$, we can assume
that $q$ does not involve $x_0$ and $x_3$.
Then $q$ divides both $x_2^3$ and~$x_1^3$, which is impossible if $q$ is not a constant.

The irreducibility of $\check{\mathrm C}$ implies its primality
since the algebra $\mathbb R[x_0,x_1,x_2,x_3]$ is a unique factorization domain.
This is why the ideal $I=(\check{\mathrm C})$ is prime.
Hence it is radical as well, i.e., it coincides with its radical
$\sqrt I:=\{g\in\mathbb R[x_0,x_1,x_2,x_3]\mid g^m\in I\text{ for some }m\in\mathbb N\}$.

Moreover, if $\deg\mathrm Q-\ord\mathrm Q=:m\in\mathbb N$,
then the polynomial~$\check{\mathrm Q}$ is a linear combination
of monomials of the form $\check{\mathrm C}^mx_3^{i_3}x_2^{i_2}x_1^{i_1}x_0^{i_0}$,
where $i_3+i_2+i_1+i_0=\ord\mathrm Q-3m=\deg\mathrm Q-4m$.
Indeed, in the light of the above arguments, the polynomial~$\check{\mathrm Q}$
is of the form $\check{\mathrm C}^lF$ for some $l\in\mathbb N$,
where $F$ is a homogeneous polynomial of the degree $\deg\mathrm Q-4l$
with $F\notin (\check{\mathrm C})$.
This implies that $l\leqslant m$.
Assuming that $l<m$, by elementary degree counting we have $\deg F>\ord F$,
which thus gives us that $F\in(\check{\mathrm C})$.
This contradiction proves the required claim.

As a result, we prove that the set $\mathcal B$ of the products listed in the theorem's statement
spans the subspace~$\Upsilon^{\ord}_n$.

Consider the linear combination
\begin{gather*}
\begin{split}
\mathrm Q:=&\sum_{j=1}^{k+1}\sum_{|i|=n-3j}\lambda_{ji_0i_1i_2i_3}\mathrm C^j
(\mathrm P^3)^{i_3}(\mathrm P^2)^{i_2}(\mathrm P^1)^{i_1}(\mathrm P^0)^{i_0}
\\&{}+\sum_{|i|\leqslant n}\lambda_{0i_0i_1i_2i_3}
(\mathrm P^3)^{i_3}(\mathrm P^2)^{i_2}(\mathrm P^1)^{i_1}(\mathrm P^0)^{i_0},
\end{split}
\end{gather*}
where $|i|:=i_0+i_1+i_2+i_3$ and $\lambda_{ji_0i_1i_2i_3}\in\mathbb R$.
Suppose that $\mathrm Q=0$. Then
\begin{gather}\label{eq:LinIndepCondition}
\psi_2(\mathrm Q)=\sum_{j=1}^{k+1}\hat{\mathrm C}^j\sum_{|i|=n-3j}\lambda_{ji_0i_1i_2i_3}x_3^{i_3}x_2^{i_2}x_1^{i_1}x_0^{i_0}
+\sum_{|i|\leqslant n}\lambda_{0i_0i_1i_2i_3}x_3^{i_3}x_2^{i_2}x_1^{i_1}x_0^{i_0}=0,
\end{gather}
where we assign $j=0$ for the terms in the last sum.
We have $\deg \hat{\mathrm C}^jx_3^{i_3}x_2^{i_2}x_1^{i_1}x_0^{i_0}=|i|+4j$.
Since all monomials in~\eqref{eq:LinIndepCondition} are different,
they are linearly independent, and thus $\lambda_{ji_0i_1i_2i_3}=0$ for all relevant values of~$(j,i_0,i_1,i_2,i_3)$.
We obtain that the set $\mathcal B$ is linearly independent.
Therefore, it is a basis of the subspace~$\Upsilon^{\ord}_n$.
\end{proof}

\begin{corollary}\label{cor:Upsilon_nDim}
The dimension of the subspace $\Upsilon^{\ord}_n$ of the algebra~$\Upsilon_{\mathfrak r}$,
which consists of differential operators of order less than or equal to $n$, is
\begin{gather}\label{eq:DimOperOrderLessEqN}
\dim\Upsilon^{\ord}_n=
\begin{cases}
\frac1{18}(n+1)(n+3)(n^2+4n+6)
&\mbox{if}\quad n\equiv0~\mbox{or}~2\pmod3,
\\[.5ex]
\frac1{18}(n+2)^2(n^2+4n+5)
&\mbox{if}\quad n\equiv1\pmod3.
\end{cases}
\end{gather}
The dimension of the quotient space $\Upsilon^{\ord}_n/\Upsilon^{\ord}_{n-1}$
associated with the $n$th order differential operators in the algebra $\Upsilon_{\mathfrak r}$ is
\begin{gather*}
\dim\Upsilon^{\ord}_n/\Upsilon^{\ord}_{n-1}=
\begin{cases}
\frac19(2n+3)(n^2+3n+3)
&\mbox{if}\quad n\equiv0\pmod3,
\\[.5ex]
\frac19(n+2)(2n^2+5n+5)
&\mbox{if}\quad n\equiv1\pmod3,
\\[.5ex]
\frac19(n+1)(2n^2+7n+8)
&\mbox{if}\quad n\equiv2\pmod3.
\end{cases}
\end{gather*}
\end{corollary}

\begin{proof}
In view of Theorem~\ref{thm:BasisOforderNOperators},
the dimension of the space $\Upsilon^{\ord}_n$ is
\begin{gather*}
\dim\Upsilon^{\ord}_n
=\sum_{k=0}^n\binom{k+3}3+\sum_{k=1}^{\lfloor n/3\rfloor}\binom{n-3(k-1)}3,
\end{gather*}
where $\lfloor x\rfloor$ denotes the ``floor'' function.
By the induction with respect to the parameter $n\in\mathbb N_0$,
one can show that the above sum coincides with the value given in~\eqref{eq:DimOperOrderLessEqN}.

Since $\Upsilon^{\ord}_{n-1}\subset\Upsilon^{\ord}_n$,
we have $\dim\Upsilon^{\ord}_n/\Upsilon^{\ord}_{n-1}=\dim\Upsilon^{\ord}_n-\dim\Upsilon^{\ord}_{n-1}$.
\end{proof}

\section{Polynomial solutions}\label{sec:RemarkableFPPolynomialSols}

Given a linear system of differential equations~$\mathcal L$
and its recursion operator~$\mathrm Q$ that is a linear differential operator in total derivatives
with coefficients depending only on the system's independent variables,
we call $\mathrm Q$ a \emph{linear differential recursion operator} of~$\mathcal L$.
Further, the function~$\mathrm Qh$ is a solution of~$\mathcal L$ whenever the function~$h$ is~\cite{kova2024a},
see~\cite{shte1987a,shte1989a} for first examples of generating solutions of linear differential equations
using this approach.
The operators~$\mathrm P^2$ and~$\mathrm P^3$ are recursion operators%
\footnote{%
Since the equation~\eqref{eq:RemarkableFP} is a linear partial differential equation, it is not a surprise
that its basic recursion operators~$\mathrm P^0$, $\mathrm P^1$, $\mathrm P^2$ and $\mathrm P^3$
(see Lemma~\ref{lem:RemarkableFPLinGenSyms} below)
are usual linear differential operators in total derivatives,
but not pseudodifferential ones as commonly happens for nonlinear differential equations.
This is why it suffices to use here the formal interpretation of recursion operators
in the sense of~\cite{olve1977a} and \cite[Definition~5.20]{olve1993A}
and not to involve the more advanced interpretation of them as certain B\"acklund transformations.
The later interpretation was suggested in~\cite{papa1990a}, more explicitly formulated in~\cite{guth1994a},
further developed in~\cite{marv1996a} and intensively applied later
\cite{jano2024a,marv2004b,marv2003a,serg2017a,serg2022a}.
}
of the equation~\eqref{eq:RemarkableFP}
and~$u=1$ is its solution.
Hence this equation possesses the solutions $(\mathrm P^3)^k(\mathrm P^2)^l1$, $k,l\in\mathbb N_0$,
which are polynomials of $(t,x,y)$ and are linearly independent.
Moreover, as the following lemma states,
these solutions exhaust, up to linearly combining them,
all solutions of this equation that are polynomial with respect to~$x$.
\looseness=-1

\begin{lemma}\label{lem:RemarkableFPPolynomialSols}
The space~$\mathscr P_n$ of solutions of the remarkable Fokker--Planck equation~\eqref{eq:RemarkableFP}
that are polynomials with respect to~$x$ of degree less than or equal to~$n\in\mathbb N_0$
with coefficients depending on $(t,y)$ is of dimension $(n+1)(n+2)/2$.
All of its elements are polynomial with respect to the entire tuple of independent variables $(t,x,y)$
and it admits a basis consisting of the polynomials $(\mathrm P^3)^k(\mathrm P^2)^l1$, $0\leqslant k+l\leqslant n$.
\end{lemma}

\begin{proof}
Substituting the general form $u=\sum_{j=0}^nf^j(t,y)x^j$
of polynomials with respect to~$x$ of degree less than or equal to~$n\in\mathbb N_0$
into the equation~\eqref{eq:RemarkableFP} and splitting with respect to~$x$, we derive the system
\[
\Delta_j\colon\quad f^j_t+f^{j-1}_y=(j+1)(j+2)f^{j+2},\quad j=0,\dots,n+1,
\]
where the equation~$\Delta_j$ is obtained by collecting coefficients of~$x^j$,
and we assume that $f^j=0$ if $j<0$ or $j>n$.
The equations~$\Delta_{n+1}$ and~$\p_y\Delta_n$ take the form $f^n_y=0$ and $f^{n-1}_{yy}=0$, respectively.
Continuing by the induction with respect to~$j$ down to $j=1$
with the differential consequences $\p_y^{n-j+1}\Delta_j$,
we obtain that $\p_y^{n-j+1}f^j=0$, $j=0,\dots,n$,
i.e., $f^j$ is a polynomial with respect to~$y$ of degree less than or equal to~$n-j$
with coefficients depending on $t$.
More specifically, the equations~$\Delta_{n+1}$, $\Delta_n$, $\Delta_{n-1}$, $\Delta_j$, $j=n-2,\dots,2,1$,
respectively take the form
\begin{gather*}
f^n_y=0,\quad
f^{n-1}_y=-f^n_t,\quad
f^{n-2}_y=-f^{n-1}_t,\\
f^j_y=-f^{j+1}_t+(j+2)(j+3)f^{j+3},\quad j=n-3,\dots,1,0.
\end{gather*}
Therefore,
$f^n=\tilde f^n(t)$,
$f^{n-1}=-\tilde f^n_t(t)y+\tilde f^{n-1}(t)$,
$f^{n-2}=\frac12\tilde f^n_{tt}(t)y^2-\tilde f^{n-1}_t(t)y+\tilde f^{n-2}(t)$.
In general, $\tilde f^j$ denotes the coefficient of~$y^0$ in~$f^j$.
By the induction with respect to~$j$ down to $j=0$,
we can show that the coefficients of~$y^{n-j}$ and~$y^{n-j-1}$ in $f^j$
are equal to $(-1)^{n-j}\p_t^{n-j}\tilde f^n/(n-j)!$ and $(-1)^{n-j-1}\p_t^{n-j-1}\tilde f^{n-1}/(n-j-1)!$,
respectively.
Moreover, the other coefficients of~$f^j$ as a polynomial in~$y$,
except the zero-degree coefficient~$\tilde f^j$,
are expressed in terms of derivatives of~$\tilde f^i$, $i>j$, with respect to~$t$.
Then the equation~$\Delta_0$: $f^0_t=2f^2$ implies
$\p_t^{n+1}\tilde f^n=0$, $\p_t^n\tilde f^{n-1}=0$
and $\p_t^{n-j+1}\tilde f^{n-j}=g^{n-j}$, $j=2,\dots,n$,
where $g^{n-j}$ is a polynomial in~$t$ expressed in terms of derivatives of~$\tilde f^{n-i}$, $i<j$, with respect to~$t$.
The dimension of the solution space of the system for~$\tilde f^j$, $j=0,\dots,n$,
is $(n+1)(n+2)/2$ and coincides with $\dim\mathscr P_n$.

The polynomial solutions $(\mathrm P^3)^k(\mathrm P^2)^l1$, $0\leqslant k+l\leqslant n$,
of the equation~\eqref{eq:RemarkableFP} are linearly independent.
Their number is equal to $(n+1)(n+2)/2$ as well.
Therefore, these polynomials constitute a basis of~$\mathscr P_n$.
\end{proof}

\begin{lemma}\label{lem:SolInhomEquation}
A particular solution of the inhomogeneous equation
$\mathrm Fu=t^r(\mathrm P^3)^i(\mathrm P^2)^j1$,
where $\mathrm F:=\mathrm D_t+x\mathrm D_y-\mathrm D_x^{\,\,2}$ and $i,j,r\in\mathbb N_0$,
is $u=(r+1)^{-1}t^{r+1}(\mathrm P^3)^i(\mathrm P^2)^j1$.
\end{lemma}

\begin{proof}
Since $u=h:=(\mathrm P^3)^i(\mathrm P^2)^j1$
is a solution of the homogeneous counterpart~\eqref{eq:RemarkableFP}
of the equation to be solved, $\mathrm Fh=0$, we obtain
$\mathrm F\big((r+1)^{-1}t^{r+1}h\big)=t^rh+(r+1)^{-1}t^{r+1}\mathrm Fh=t^rh$.
\end{proof}

\section{Generalized symmetries}\label{sec:RemarkableFPGenSyms}

Hereafter, we use the following notation.
The jet variable~$u_{kl}$ is identified with the derivative $\p^{k+l}u/\p x^k\p y^l$, $k,l\in\mathbb N_0$.
In particular, $u_{00}:=u$.
The jet variables $(t,x,y,u_{kl},k,l\in\mathbb N_0)$
constitute the standard coordinates on the manifold~$\mathscr F$ defined by the equation~\eqref{eq:RemarkableFP}
and its differential consequences in the infinite-order jet space $\mathrm J^\infty(\mathbb R^3_{t,x,y}\times\mathbb R_u)$
with the independent variables $(t,x,y)$ and the dependent variable~$u$.
We consider differential functions defined on~$\mathscr F$, and
$\eta[u]$ denotes a differential function~$\eta$ of~$u$
that depends on $t$, $x$, $y$ and a finite number of $u_{kl}$.
Recall that the order~$\ord\eta[u]$ of a differential function~$\eta[u]$ is the highest order of derivatives of~$u$ involved in~$\eta[u]$
if there are such derivatives, and $\ord \eta[u]=-\infty$ otherwise.
For a generalized vector field $Q=\eta[u]\p_u$, we define $\ord Q:=\ord\eta[u]$.
The restrictions of the operators $\mathrm D_t$, $\mathrm D_x$ and~$\mathrm D_y$
of total derivatives in~$t$, $x$ and~$y$ to such differential functions on~$\mathscr F$
are respectively
\begin{gather*}
\hat{\mathrm D}_t=\p_t+\sum_{k,l=0}^\infty (u_{k+2,l}-xu_{k,l+1}-ku_{k-1,l+1})\p_{u_{kl}},\\
\hat{\mathrm D}_x=\p_x+\sum_{k,l=0}^\infty u_{k+1,l}\p_{u_{kl}},\quad
\hat{\mathrm D}_y=\p_y+\sum_{k,l=0}^\infty u_{k,l+1}\p_{u_{kl}}.
\end{gather*}

As for any evolution equation,
it is natural to identify the quotient algebra of generalized symmetries of~\eqref{eq:RemarkableFP}
with respect to the equivalence of generalized symmetries
with the algebra
\[
\Sigma:=\big\{\eta[u]\p_u\mid\mathrm F\eta[u]=0\big\}
\quad\mbox{with}\quad\mathrm F:=\hat{\mathrm D}_t+x\hat{\mathrm D}_y-\hat{\mathrm D}_x^{\,\,2}
\]
of canonical representatives of equivalence classes, see~\cite[Section~5.1]{olve1993A}.
The subspace
\[
\Sigma^n:=\big\{\eta[u]\p_u\in\Sigma\mid\ord\eta[u]\leqslant n\big\},\quad n\in\mathbb N_0\cup\{-\infty\},
\]
of~$\Sigma$ is interpreted as the space of generalized symmetries of order less than or equal to~$n$.
The subspace~$\Sigma^{-\infty}$ can be identified with the subalgebra~$\mathfrak g^{\rm lin}$
of Lie symmetries of the equation~\eqref{eq:RemarkableFP}
that are associated with the linear superposition of solutions of this equation,
\[\Sigma^{-\infty}=\{\mathcal Z(h):=h(t,x,y)\p_u\mid h_t+xh_y=h_{xx}\}\simeq\mathfrak g^{\rm lin}.\]

The subspace family $\{\Sigma^n\mid n\in\mathbb N_0\cup\{-\infty\}\}$ filters the algebra~$\Sigma$.
Denote
$\Sigma^{[n]}:=\Sigma^n/\Sigma^{n-1}$, $n\in\mathbb N$,
$\Sigma^{[0]}:=\Sigma^0/\Sigma^{-\infty}$ and
$\Sigma^{[-\infty]}:=\Sigma^{-\infty}$.
The space $\Sigma^{[n]}$ is naturally identified with the space of canonical representatives
of cosets of~$\Sigma^{n-1}$ in~$\Sigma^n$
and thus assumed as the space of $n$th order generalized symmetries of the equation~\eqref{eq:RemarkableFP},
$n\in\mathbb N_0\cup\{-\infty\}$.%
\footnote{
The filtration $\Sigma=\cup_{n\in\mathbb N_0\cup\{-\infty\}}\Sigma^n$ of the algebra $\Sigma$
gives rise to the associated graded algebra
$\mathop{\rm gr}\Sigma=\oplus_{n\in\mathbb N_0}\Sigma^{[n]}$,
where $\Sigma^{[n]}:=\Sigma^n/\Sigma^{n-1}$ with $\Sigma^{-1}:=\Sigma^{-\infty}$.
In this notation, the space $\Sigma^{[n]}$ is the homogeneous component of degree $n$ of the $\mathbb N_0$-graded algebra~$\mathop{\rm gr}\Sigma$.
}

In view of the linearity of the equation~\eqref{eq:RemarkableFP},
an important subalgebra of its generalized symmetries consists of its linear generalized symmetries,
\[
\Lambda:=\bigg\{\eta[u]\p_u\in\Sigma\ \Big|\ \exists\,n\in\mathbb N_0,\,
\exists\,\eta^{kl}=\eta^{kl}(t,x,y),\, k,l\in\mathbb N_0,\, k+l\leqslant n\colon\,
\eta[u]=\!\!\sum_{k+l\leqslant n}\eta^{kl}u_{kl}\bigg\}.
\]
The subspace $\Lambda^n:=\Lambda\cap\Sigma^n$ of~$\Lambda$ with $n\in\mathbb N_0$
is constituted by the generalized symmetries with characteristics of the form
\begin{gather}\label{eq:RemarkableFPLinGenSyms}
\eta[u]=\sum_{k+l\leqslant n}\eta^{kl}(t,x,y)u_{kl}.
\end{gather}
A linear generalized symmetry is of order~$n$ if and only if
there exists a nonvanishing coefficient~$\eta^{kl}$ with $k+l=n$.
The quotient spaces $\Lambda^{[n]}=\Lambda^n/\Lambda^{n-1}$, $n\in\mathbb N$, and the subspace $\Lambda^{[0]}=\Lambda^0$
are naturally embedded in the respective spaces $\Sigma^{[n]}$, $n\in\mathbb N_0$,
when taking linear canonical representatives for cosets of~$\Sigma^{n-1}$ containing linear generalized symmetries.
We interpret the space $\Lambda^{[n]}$ as the space of $n$th order linear generalized symmetries
of the equation~\eqref{eq:RemarkableFP}, $n\in\mathbb N_0$.%\looseness=1

\begin{lemma}\label{lem:RemarkableFPLinGenSyms}
The algebra~$\Lambda$ coincides with the algebra~$\Lambda_{\mathfrak r}$ of linear generalized symmetries
generated by acting with the recursion operators~$\mathrm P^3$, $\mathrm P^2$, $\mathrm P^1$ and $\mathrm P^0$
on the elementary seed symmetry vector field $u\p_u$,
\begin{gather*}
\Lambda=\Lambda_{\mathfrak r}:=
\big\langle\big((\mathrm P^3)^{i_3}(\mathrm P^2)^{i_2}(\mathrm P^1)^{i_1}(\mathrm P^0)^{i_0}u\big)\p_u
\mid i_0,i_1,i_2,i_3\in\mathbb N_0\big\rangle.
\end{gather*}
\end{lemma}

\begin{proof}
The condition $\mathrm F\eta[u]=0$ of invariance of the equation~\eqref{eq:RemarkableFP}
with respect to linear generalized symmetries with characteristics~$\eta$ of the form~\eqref{eq:RemarkableFPLinGenSyms}
can represented as
\begin{gather*}
\noprint{
\eta^{kl}_tu_{kl}+\eta^{kl}(u_{k+2,l}-xu_{k,l+1}-ku_{k-1,l+1})
+x\eta^{kl}_yu_{kl}+x\eta^{kl}u_{k,l+1}
-\eta^{kl}_{xx}u_{kl}-2\eta^{kl}_xu_{k+1,l}-\eta^{kl}u_{k+2,l}\\
}
(\eta^{kl}_t+x\eta^{kl}_y-\eta^{kl}_{xx})u_{kl}-k\eta^{kl}u_{k-1,l+1}-2\eta^{kl}_xu_{k+1,l}=0.
%\\\qquad=u_{kl}\big(\mathrm F\eta^{kl}-(k+1)\eta^{k+1,l-1}-2\eta^{k-1,l}_x\big)=0
\end{gather*}
Splitting this condition with respect to the jet variables~$u_{kl}$,
we derive the system of determining equations for the coefficients~$\eta^{kl}$,
\[%\label{eq:RemarkableFPDetEqsForLinGenSyms}
\Delta_{kl}\colon\ \mathrm F\eta^{kl}-(k+1)\eta^{k+1,l-1}-2\eta^{k-1,l}_x=0,\quad
k,l\in\mathbb N_0,\quad k+l\leqslant n+1,
\]
where we denote $n:=\ord\eta$ and assume $\eta^{kl}=0$ if $k<0$ or $l<0$ or $k+l>n$.

\begin{figure}[t]
\centering
\begin{subfigure}[b]{0.48\textwidth}
\centering
\includegraphics[width=\textwidth]{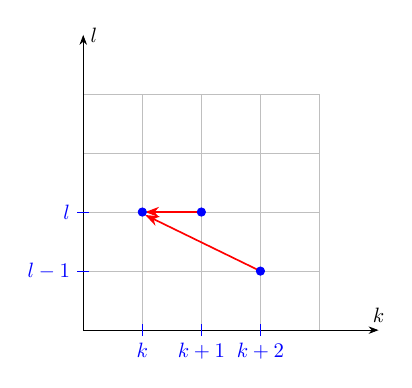}
\caption{}
\label{fig:GraphXPattern}
\end{subfigure}
\quad
\begin{subfigure}[b]{0.48\textwidth}
\centering
\includegraphics[width=\textwidth]{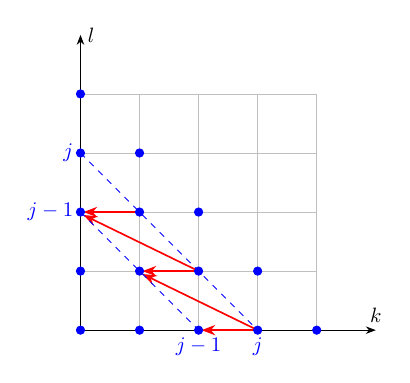}
\caption{}
\label{fig:GraphXInduction}
\end{subfigure}
\caption{The first induction (downward).
(a) Relation pattern.
(b) Induction step.}
\label{fig:GraphX}
\end{figure}

For each $(k,l)\in\mathbb N_0\times\mathbb N_0$ with $k+l\leqslant n$,
we rewrite the equation $\Delta_{k+1,l}$ as
\[
2\eta^{kl}_x=\mathrm F\eta^{k+1,l}-(k+2)\eta^{k+2,l-1}.
\]
In other words, the coefficient~$\eta^{kl}$ is defined by the coefficients~$\eta^{k+1,l}$ and~$\eta^{k+2,l-1}$
modulo a summand depending only on $(t,y)$.
Associating $\eta^{k'l'}$ with the point $(k',l')$ in the grid $\mathbb N_0\times\mathbb N_0$,
we geometrically depict this relation pattern in Figure~\ref{fig:GraphXPattern}.
Therefore, for each fixed $j\in\mathbb N$,
the coefficients~$\eta^{kl}$ with $k+l=j$ define the coefficients~$\eta^{kl}$ with $k+l=j-1$
up to summands depending only on $(t,y)$,
see Figure~\ref{fig:GraphXInduction}.
Thus, the induction with respect to~$m:=k+l$ from $m=n+1$, where $\eta^{kl}=0$, downwards to $m=0$,
in the course of which each induction step is realized
as the secondary induction with respect to~$l$ from $l=m$ downwards to $l=0$,
straightforwardly implies that the coefficient~$\eta^{kl}$ is a polynomial with respect to~$x$
of the degree at most $2n-2(k+l)$ with coefficients depending on $(t,y)$.

\begin{figure}[t]
\centering
\begin{subfigure}[b]{0.48\textwidth}
\centering
\includegraphics[width=\textwidth]{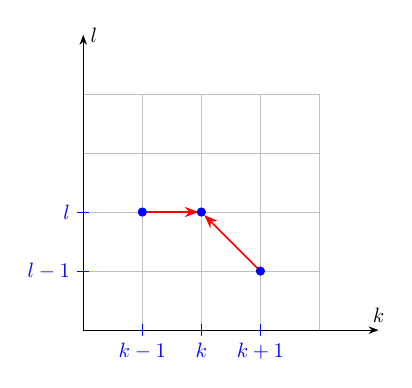}
\caption{}
\label{fig:GraphPattern}
\end{subfigure}
\quad
\begin{subfigure}[b]{0.48\textwidth}
\centering
\includegraphics[width=\textwidth]{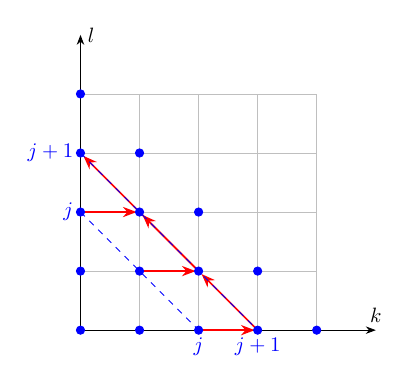}
\caption{}
\label{fig:GraphInduction}
\end{subfigure}
\caption{The second induction (upward).
(a) Relation pattern.
(b) Induction step.}
\label{fig:Graph}
\end{figure}

Now we prove that $\eta^{kl}\in\mathscr T$ for any $k,l\in\mathbb N_0$,
where $\mathscr T$ is the space of finite linear combinations of terms $t^r(\mathrm P^3)^i(\mathrm P^2)^j1$, $i,j,r\in\mathbb N_0$.
Using Lemma~\ref{lem:SolInhomEquation},
we carry out the induction with respect to~$m:=k+l$ in the opposite direction, from $m=0$ upwards to $m=n$,
as shown in Figure~\ref{fig:GraphInduction},
where each induction step is performed as the secondary induction with respect to~$l$ from ${l=0}$ upwards to $l=m$.
The induction base $k=l=0$ follows in view of Lemma~\ref{lem:RemarkableFPPolynomialSols}
from the equation $\Delta_{00}$: $\mathrm F\eta^{00}=0$ and the polynomiality of~$\eta^{00}$ with respect to~$x$.
On the step $(k,l)$, we have $\eta^{k+1,l-1},\eta^{k-1,l}\in\mathscr T$ by the induction supposition.
Taking into account $[\mathrm D_x,\mathrm P^2]=1$ and ${[\mathrm D_x,\mathrm P^3]=3t}$, we obtain
\[
\big(t^r(\mathrm P^3)^i(\mathrm P^2)^j1\big)_x=
3it^{r+1}(\mathrm P^3)^{i-1}(\mathrm P^2)^j1+jt^r(\mathrm P^3)^i(\mathrm P^2)^{j-1}1.
\]
Therefore, $\eta^{k-1,l}_x\in\mathscr T$ as well.
Considering $\Delta_{kl}$ as an inhomogeneous equation with respect to~$\eta^{kl}$,
we represent $\eta^{kl}$ as the sum
of a particular solution~$\hat\eta^{kl}$ of this equation according to Lemma~\ref{lem:SolInhomEquation}
and a solution~$\check\eta^{kl}$ of the homogeneous counterpart $\mathrm F\eta^{kl}=0$ of the equation $\Delta_{kl}$,
see Figure~\ref{fig:GraphPattern} for an illustration.
Since $\hat\eta^{kl}\in\mathscr T$ due to the choice in Lemma~\ref{lem:SolInhomEquation}
and $\eta^{kl}$ is polynomial with respect to~$x$ in view of the above arguments,
$\check\eta^{kl}$ is also polynomial with respect to~$x$ and
Lemma~\ref{lem:RemarkableFPPolynomialSols} implies that $\check\eta^{kl}\in\mathscr T$,
including only terms with $r=0$.
Hence $\eta^{kl}=\hat\eta^{kl}+\check\eta^{kl}\in\mathscr T$.

As a result,
we derive the following representation for~$\eta$:
\begin{gather}\label{eq:RemarkableFPGenSymRepresentation}
\eta=\sum_{i,j,k,l\in\mathbb N_0}c_{ijkl}W^{ijkl}, \quad
W^{ijkl}:=\big((\mathrm P^3)^i(\mathrm P^2)^j1\big)u_{kl}+\sum_{(k',l')\succ(k,l)}V^{ijklk'l'}u_{k'l'},
\end{gather}
where
$(k',l')\succ(k,l)$ means that $k',l'\in\mathbb N_0$, $l'\geqslant l$, $k'+l'\geqslant k+l$ and $(k',l')\ne(k,l)$,
each $V^{ijklk'l'}$ is an element of~$\mathscr T$ that is completely defined by $(i,j,k,l,k',l')$,
has $r>0$ for each of its summand,
and only finite number of~$c_{ijkl}$ and of~$V^{ijklk'l'}$ are nonzero.
In other words,
any generalized symmetry $\eta\p_u$ of the equation~\eqref{eq:RemarkableFP}
is completely defined by the corresponding coefficients $c_{ijkl}$ of~$W^{ijkl}$
or, equivalently, of $(\mathrm P^3)^i(\mathrm P^2)^ju_{kl}$
in its representation~\eqref{eq:RemarkableFPGenSymRepresentation}.
At the same time,
\[
(\mathrm P^3)^i(\mathrm P^2)^j(\mathrm P^1)^k(\mathrm P^0)^lu=
\big((\mathrm P^3)^i(\mathrm P^2)^j1\big)u_{kl}+\sum_{(k',l')\succ(k,l)}\tilde V^{ijklk'l'}u_{k'l'},
\]
where the coefficients~$\tilde V^{ijklk'l'}$ have the same properties as $V^{ijklk'l'}$.
Therefore, $\eta\p_u\in\Lambda_{\mathfrak r}$, i.e., $\Lambda\subseteq\Lambda_{\mathfrak r}$.
The inverse inclusion follows from the definitions of $\Lambda$ and~$\Lambda_{\mathfrak r}$.
Thus, $\Lambda=\Lambda_{\mathfrak r}$.
\end{proof}

\begin{corollary}\label{cor:Upsilon}
The associative algebra $\Upsilon$ of linear differential recursion operators
of the equation~\eqref{eq:RemarkableFP} coincides with the algebra~$\Upsilon_{\mathfrak r}$.
\end{corollary}

\begin{corollary}\label{cor:LambdaIso1}
The algebra~$\Lambda=\Lambda_{\mathfrak r}$ is anti-isomorphic to the algebra~$\Upsilon_{\mathfrak r}^{\mbox{\tiny$(-)$}}$
and, therefore, to the Lie algebra
associated with the quotient of the universal enveloping algebra of the Lie algebra~$\mathfrak r$
by the principal ideal $(\iota(\mathcal I)+1)$ generated by $\iota(\mathcal I)+1$,
$\Lambda_{\mathfrak r}\simeq\big(\mathfrak U(\mathfrak r)/(\iota(\mathcal I)+1)\big)^{\mbox{\tiny$(-)$}}$.
\end{corollary}

\begin{proof}
The correspondence
$\big((\mathrm P^3)^{i_3}(\mathrm P^2)^{i_2}(\mathrm P^1)^{i_1}(\mathrm P^0)^{i_0}u\big)\p_u\mapsto
(\mathrm P^3)^{i_3}(\mathrm P^2)^{i_2}(\mathrm P^1)^{i_1}(\mathrm P^0)^{i_0}$
extended by linearity
straightforwardly gives us a vector-space isomorphism $\varphi$
from $\Lambda_{\mathfrak r}$ to $\Upsilon_{\mathfrak r}$.
Consider operators $\mathrm Q,\mathrm R\in\Upsilon_{\mathfrak r}$, i.e.,
$\mathrm Q=Q^{ij}\mathrm D_x^i\mathrm D_y^j$ and $\mathrm R=R^{ij}\mathrm D_x^i\mathrm D_y^j$,
where only a finite number of the polynomials~$Q^{ij}$ and~$R^{ij}$ of $(t,x,y)$ are nonzero.
Here and in what follows we assume summation with respect two repeated indices~$i$ and~$j$ through~$\mathbb N_0$.
In view of \cite[Proposition~5.15]{olve1993A},
the commutator $\big[(\mathrm Qu)\p_u,(\mathrm Ru)\p_u\big]$ of evolutionary generalized vector fields $(\mathrm Qu)\p_u$ and $(\mathrm Ru)\p_u$
from~$\Lambda_{\mathfrak r}$
is an evolutionary vector field with characteristic
\begin{gather*}
\mathop{\rm pr}\big((\mathrm Qu)\p_u\big)(\mathrm Ru)-\mathop{\rm pr}\big((\mathrm Ru)\p_u\big)(\mathrm Qu)
=\mathrm D_x^i\mathrm D_y^j(\mathrm Qu)\p_{u_{ij}}(\mathrm Ru)-\mathrm D_x^i\mathrm D_y^j(\mathrm Ru)\p_{u_{ij}}(\mathrm Qu)
\\
\qquad
=R^{ij}\mathrm D_x^i\mathrm D_y^j(\mathrm Qu)-Q^{ij}\mathrm D_x^i\mathrm D_y^j(\mathrm Ru)
=\mathrm R(\mathrm Qu)-\mathrm Q(\mathrm Ru)=[\mathrm R,\mathrm Q]u,
\end{gather*}
where $\mathop{\rm pr}(\eta\p_u)$ denotes the prolongation of a generalized vector field $\eta\p_u$ with respect~$x$ and~$y$,
$\mathop{\rm pr}(\eta\p_u)=(\mathrm D_x^i\mathrm D_y^j\eta)\p_{u_{ij}}$.
Therefore, $\varphi([\mathrm Qu\p_u,\mathrm Ru\p_u])=-[\mathrm Q,\mathrm R]$,
i.e., $\varphi\colon\Lambda_{\mathfrak r}\to\Upsilon_{\mathfrak r}^{\mbox{\tiny$(-)$}}$ is an anti-isomorphism,
which combines with Lemma~\ref{lem:IsoToUnivEnv} to the second assertion in this theorem.
\end{proof}

We can reformulate Corollary~\ref{cor:LambdaIso1},
recalling the isomorphism of $\mathfrak r$ to the rank-two Heisenberg algebra
${\rm h}(2,\mathbb R)=\langle p_1,p_2,q_1,q_2,c\rangle$, Remark~\ref{rem:WeylAlgebra}
and Corollary~\ref{cor:IsoToWeyl}.
In particular,
\[
\mathfrak U(\mathfrak r)/(\iota(\mathcal I)+1)
\simeq\mathfrak U\big({\rm h}(2,\mathbb R)\big)/(c+1)
\simeq {\rm W}(2,\mathbb R)^{\rm op}\simeq {\rm W}(2,\mathbb R).
\]

\begin{corollary}
The algebra~$\Lambda=\Lambda_{\mathfrak r}$ of the linear generalized symmetries
of the remarkable Fokker--Planck equation~\eqref{eq:RemarkableFP}
is isomorphic to the Lie algebra ${\rm W}(2,\mathbb R)^{\mbox{\tiny$(-)$}}$
associated with the second Weyl algebra ${\rm W}(2,\mathbb R)$,
$\Lambda\simeq {\rm W}(2,\mathbb R)^{\mbox{\tiny$(-)$}}$.
\end{corollary}

Hence the center of the algebra~$\Lambda$ is one-dimensional and spanned by $u\p_u$,
and the quotient algebra $\Lambda/\langle u\p_u\rangle$ is simple.

Combining Corollaries~\ref{cor:Upsilon_nDim} and~\ref{cor:LambdaIso1},
we derive that
\begin{gather*}%\label{eq:Lambda_nDim}
\dim\Lambda^n=
\begin{cases}
\frac1{18}(n+1)(n+3)(n^2+4n+6)
&\mbox{if}\quad n\equiv0~\mbox{or}~2\pmod3,
\\[.5ex]
\frac1{18}(n+2)^2(n^2+4n+5)
&\mbox{if}\quad n\equiv1\pmod3,
\end{cases}
\\[1ex]
\dim\Lambda^{[n]}=
\begin{cases}
\frac19(2n+3)(n^2+3n+3)
&\mbox{if}\quad n\equiv0\pmod3,
\\[.5ex]
\frac19(n+2)(2n^2+5n+5)
&\mbox{if}\quad n\equiv1\pmod3,
\\[.5ex]
\frac19(n+1)(2n^2+7n+8)
&\mbox{if}\quad n\equiv2\pmod3.
\end{cases}
\end{gather*}

\begin{theorem}\label{thm:GenSymsRemarkableFP}
The algebra of canonical representatives of generalized symmetries
of the remarkable Fokker--Planck equation~\eqref{eq:RemarkableFP} is
$\Sigma=\Lambda_{\mathfrak r}\lsemioplus\Sigma^{-\infty}$,
where
\begin{gather*}
\Lambda_{\mathfrak r}=\big\langle\big((\mathrm P^3)^{i_3}(\mathrm P^2)^{i_2}(\mathrm P^1)^{i_1}(\mathrm P^0)^{i_0}u\big)\p_u
\mid i_0,i_1,i_2,i_3\in\mathbb N_0\big\rangle,\quad
\Sigma^{-\infty}:=\big\{\mathfrak Z(h)\big\}.
\end{gather*}
Here the parameter function~$h$ runs through the solution set of~\eqref{eq:RemarkableFP}.
\end{theorem}

\begin{proof}
Lemma~\ref{lem:RemarkableFPLinGenSyms} obviously implies that $\dim\Lambda^{[n]}<\infty$ for any $n\in\mathbb N_0$,
and thus $\Sigma^{[n]}=\Lambda^{[n]}$ for any $n\in\mathbb N_0$.
The last implication is just a particular formulation of the Shapovalov--Shirokov theorem~\cite[Theorem~4.1]{shap1992a}
for the equation~\eqref{eq:RemarkableFP}.
Therefore, $\Sigma=\Lambda\lsemioplus\Sigma^{-\infty}=\Lambda_{\mathfrak r}\lsemioplus\Sigma^{-\infty}$.
\end{proof}

In other words, the  algebra~$\Sigma$ splits over the infinite-dimensional abelian ideal~$\Sigma^{-\infty}$
of trivial generalized symmetries associated with the linear superposition of solutions.
The complementary subalgebra to~$\Sigma^{-\infty}$ in~$\Sigma$,
which is naturally called the \emph{essential algebra of generalized symmetries},
is just the algebra $\Lambda=\Lambda_{\mathfrak r}$ of linear generalized symmetries,
which is isomorphic to the Lie algebra ${\rm W}(2,\mathbb R)^{\mbox{\tiny$(-)$}}$
associated with the second Weyl algebra ${\rm W}(2,\mathbb R)$.

\begin{remark}
The subspaces $\Lambda^{[1]}$ and $\Lambda^{[2]}$ are in fact subalgebras of~$\Lambda$.
As Lie algebras, they are isomorphic to the (nil)radical~$\mathfrak r$ of
the essential Lie invariance algebra~$\mathfrak g^{\rm ess}$ of~\eqref{eq:RemarkableFP}
and the algebra~$\mathfrak g^{\rm ess}$ itself, respectively.
An interesting question is whether the algebra~$\Lambda$ contains
finite-dimensional noncommutative subalgebras
that are ${\rm Aut}(\Lambda)$-inequivalent to subalgebras of~$\Lambda^{[2]}$.
\end{remark}

\begin{remark}
Another concept related to generating generalized symmetries similarly to recursion operators
is that of master symmetry.
There are various (in general, inequivalent) notions of master symmetries in the literature.
According to the definition presented in \cite[p.~315]{olve1993A},
a \emph{master symmetry} is a generalized (or even nonlocal) vector field $\mathfrak M$ with the property
that whenever $Q$ is a generalized symmetry of the system of evolution equations under consideration,
so is the Lie bracket $[\mathfrak M,Q]$.
Since any generalized symmetry of the system satisfies this property, it is then a master symmetry.
Mimicking techniques from~\cite{shap1992a}, it is possible to prove that within the framework of this definition
and in the setting of the Shapovalov--Shirokov theorem,
\emph{a generalized vector field is a master symmetry of a linear system of differential equations if and only if
it is a generalized symmetry of this system}.
But even in \cite[p.~317]{olve1993A}, a different notion of master symmetry
(see, e.g., \cite[Definition~2]{wang2015b} for a precise formulation of the corresponding definition)
is implicitly used when considering the potential Burgers equation as an example,
where the commutators of the presented master symmetry
only with the elements of a (proper) abelian subalgebra of the entire algebra of generalized symmetries of this equation
are such symmetries.
The analogous master symmetry for the Burgers equation is discussed, e.g., in \cite[Section~1]{wang2015b}.
The generalized vector field $yu_{yy}\p_u$, which is the ``time-independent part''
of the generalized symmetry $-\frac13\big(\mathrm P^3(\mathrm P^0)^2u\big)\p_u$ of the equation~\eqref{eq:RemarkableFP},
satisfies all the properties of master symmetries according to \cite[Definition~2]{wang2015b},
except that the right-hand side of~\eqref{eq:RemarkableFP} does not arise
as the characteristic of a generalized symmetry from the generated hierarchy.
\end{remark}

\begin{remark}\label{rem:LamdbaTwo-Generation}
The algebra~$\Lambda$ has another remarkable generation property
following from the analogous property of~${\rm W}(2,\mathbb R)^{\mbox{\tiny$(-)$}}$~\cite{kova2025a}.
It is two-generated as a Lie algebra.
More specifically, any of its elements can be represented as a linear combination
of successive commutators (aka nonassociative monomials) of its two fixed elements,
e.g., those associated with the operators
\[
\mathrm P^3\mathrm P^0+5\mathrm P^2\mathrm P^1\quad\mbox{and}\quad
\big((\mathrm P^3)^3+(\mathrm P^0)^2+1\big)\big((\mathrm P^2)^3+(\mathrm P^1)^2+1\big).
\]
Similarly, both (isomorphic) algebras~$\Upsilon_{\mathfrak r}$ and~${\rm W}(2,\mathbb R)$
are two-generated as associative algebras.
\end{remark}

\begin{remark}\label{rem:W2R-Grading}
The Lie algebra ${\rm W}(2,\mathbb R)^{\mbox{\tiny$(-)$}}$ has a natural $\mathbb Z$-grading
associated with the specific inner semisimple derivation
given by the adjoint action of the element $s:=\hat p_1\hat q_1+\hat p_2\hat q_2$,
\[
{\rm ad}_s(\hat q^\kappa\hat p^\lambda)=(\kappa_1-\lambda_1+\kappa_2-\lambda_2)\hat q^\kappa\hat p^\lambda,
\]
where the elements $\hat q^\kappa\hat p^\lambda:=\hat q_1^{\kappa_1}\hat q_2^{\kappa_2}\hat p_1^{\lambda_1}\hat p_2^{\lambda_2}$,
$\kappa:=(\kappa_1,\kappa_2),\lambda:=(\lambda_1,\lambda_2)\in\mathbb N_0^2$,
constitute the canonical basis of ${\rm W}(2,\mathbb R)$, see Remark~\ref{rem:WeylAlgebra} and the commutation relation~\eqref{eq:CommutatorsWn-}.
Denote by $\Gamma_m$ the eigenspace of~${\rm ad}_s$ associated with the eigenvalue $m\in\mathbb Z$,
\[
\Gamma_m:=\big\langle\hat q^\kappa\hat p^\lambda\mid\kappa,\lambda\in\mathbb N_0^2,\,\kappa_1-\lambda_1+\kappa_2-\lambda_2=m\big\rangle.
\]
Then the algebra ${\rm W}(2,\mathbb R)$ is the direct sum of~$\Gamma_m$, $m\in\mathbb Z$.
Since ${\rm ad}_s$ is a derivation of the  Lie algebra~${\rm W}(2,\mathbb R)^{\mbox{\tiny$(-)$}}$,
this direct sum decomposition is in fact its $\mathbb Z$-grading,
\[
{\rm W}(2,\mathbb R)^{\mbox{\tiny$(-)$}}=\bigoplus_{m\in\mathbb Z}\Gamma_m
\quad\mbox{with}\quad [\Gamma_m,\Gamma_{m'}]\subseteq\Gamma_{m+m'}.
\]
The isomorphism between~$\Lambda$ and~${\rm W}(2,\mathbb R)^{\mbox{\tiny$(-)$}}$
transfers this grading to the algebra $\Lambda$,
where the role of~$s$ is played by %the generalized vector field
$\mathfrak R:=\big(\frac13\mathrm P^3\mathrm P^0u+\mathrm P^2\mathrm P^1u\big)\p_u$
and the grading components are given by
\[
\big\langle\big((\mathrm P^3)^{i_3}(\mathrm P^2)^{i_2}(\mathrm P^1)^{i_1}(\mathrm P^0)^{i_0}u\big)\p_u\mid
i_3,i_2,i_1,i_0\in\mathbb N_0,\, i_3-i_0+i_2-i_1=m\big\rangle,\quad
m\in\mathbb Z.
\]
\end{remark}

\begin{remark}\label{rem:HeatEqGenSyms}
The structure of the algebra~$\Sigma_{\rm h}$ of generalized symmetries of the linear (1+1)-dimensional heat equation
\begin{gather}\label{eq:HeatEq}
u_t=u_{xx}
\end{gather}
is similar to that of the algebra~$\Sigma$.
Indeed, the algebra $\Sigma_{\rm h}$
splits over its infinite-dimensional ideal $\Sigma_{\rm h}^{-\infty}$
associated with the linear superposition of solutions of~\eqref{eq:HeatEq},
$\Sigma_{\rm h}=\Sigma_{\rm h}^{\rm ess}\lsemioplus\Sigma_{\rm h}^{-\infty}$.
The complementary subalgebra $\Sigma_{\rm h}^{\rm ess}$ to the ideal~$\Sigma_{\rm h}^{-\infty}$
in the algebra~$\Sigma_{\rm h}$ coincides with the algebra~$\Lambda_{\rm h}$
of linear generalized symmetries of~\eqref{eq:HeatEq}, see~\cite{kova2023b}.
In view of \cite[Corollary~21]{kova2023b},
it is anti-isomorphic to the Lie algebra arising from
the quotient of the universal enveloping algebra~$\mathfrak U\big({\rm h}(1,\mathbb R)\big)$
of the rank-one Heisenberg algebra ${\rm h}(1,\mathbb R)$
by the principal two-sided ideal $(c+1)$ generated by~$c+1$,
\[\Lambda_{\rm h}\simeq\big(\mathfrak U\big({\rm h}(1,\mathbb R)\big)/(c+1)\big)^{\mbox{\tiny$(-)$}},\]
see Remark~\ref{rem:WeylAlgebra}.
Hence it is isomorphic to the Lie algebra ${\rm W}(1,\mathbb R)^{\mbox{\tiny$(-)$}}$
associated with the first Weyl algebra ${\rm W}(1,\mathbb R)$.
\end{remark}

Since any ultraparabolic linear second-order partial differential equation with three independent variables
whose essential Lie invariance algebra is eight-dimensional
is similar to the remarkable Fokker--Planck equation~\eqref{eq:RemarkableFP}
modulo point transformations, we obtain the following assertion.

\begin{corollary}\label{cor:GenSymsEquivRemarkableFP}
The algebra of canonical representatives of generalized symmetries
of any ultraparabolic linear second-order partial differential equation~$\mathcal L$
with the three independent variables $(t,x,y)$ and the dependent variable~$u$
whose essential Lie invariance algebra is eight-dimensional
is $\tilde\Sigma=\tilde\Lambda\lsemioplus\tilde\Sigma^{-\infty}$,
where
\begin{gather*}
\tilde\Lambda=\big\langle\big(
(\tilde{\mathrm P}^3)^{i_3}(\tilde{\mathrm P}^2)^{i_2}(\tilde{\mathrm P}^1)^{i_1}(\tilde{\mathrm P}^0)^{i_0}
u\big)\p_u\mid i_0,i_1,i_2,i_3\in\mathbb N_0\big\rangle,\quad
\tilde\Sigma^{-\infty}:=\big\{h(t,x,y)u\p_u\big\},
\end{gather*}
$\tilde{\mathrm P}^0$, \dots, $\tilde{\mathrm P}^3$ are Lie-symmetry operators of~$\mathcal L$
that are associated with (additional to $u\p_u$) basis elements of the radical
of the essential Lie invariance algebra of~$\mathcal L$,
and the parameter function~$h$ runs through the solution set of~$\mathcal L$.
\end{corollary}

Thus, the algebra~$\tilde\Sigma$ splits over
the algebra $\tilde\Sigma^{-\infty}$ of Lie symmetries of the equation~$\mathcal L$
related to the linear superposition of solutions of this equation,
which is the maximal abelian ideal in~$\tilde\Sigma$.
Its subalgebra complement in~$\tilde\Sigma$ is
the algebra $\tilde\Lambda$ of canonical representatives of linear generalized symmetries of~$\mathcal L$,
which is isomorphic to the Lie algebra ${\rm W}(2,\mathbb R)^{\mbox{\tiny$(-)$}}$
associated with the second Weyl algebra ${\rm W}(2,\mathbb R)$.

Specific ultraparabolic linear second-order partial differential equation with three independent variables
whose essential Lie invariance algebra is eight-dimensional
appeared in the literature independently of the equation~\eqref{eq:RemarkableFP}.
In particular, these are
the Kolmogorov backward equation~\cite{zhan2020a}\looseness=-1
\begin{gather}\label{eq:Power5FP}
u_t+xu_y=x^5u_{xx}
\end{gather}
and the Kramers equations~\cite{spic1999a}
\begin{gather*}
u_t+xu_y=\gamma u_{xx}+\gamma(x-\tfrac34\gamma y)u_x+\gamma u,\\
u_t+xu_y=\gamma u_{xx}+\gamma(x+\tfrac3{16}\gamma y)u_x+\gamma u.
\end{gather*}
Point transformations mapping these equations to the equation~\eqref{eq:RemarkableFP}
were found in~\cite[Eq.~(6)]{kova2024c} and~\cite[Section~9]{kova2023a}, respectively.
The pushforwards~$\tilde{\mathrm P}^0$, \dots, $\tilde{\mathrm P}^3$
of ~$\mathrm P^0$, \dots, $\mathrm P^3$ by these transformations take the form
\begin{gather*}
\tilde{\mathrm P}^0=\mathrm D_t,\quad
\tilde{\mathrm P}^1=y\mathrm D_t-x^2\mathrm D_x+x,\\
\tilde{\mathrm P}^2=y^2\mathrm D_t-2x^2y\mathrm D_x+2yx+x^{-1},\quad
\tilde{\mathrm P}^3=y^3\mathrm D_t-3x^2y^2\mathrm D_x+3(xy^2-t+x^{-1}y),
\\[1.5ex]
\tilde{\mathrm P}^0={\rm e}^{-\frac32\gamma t}\left(\tfrac1\gamma\mathrm D_y-\tfrac32\mathrm D_x\right),\quad
\tilde{\mathrm P}^1={\rm e}^{-\frac12\gamma t}\left(\tfrac1\gamma\mathrm D_y-\tfrac12\mathrm D_x
-\tfrac12(x+\tfrac32\gamma y)\right),\\
\tilde{\mathrm P}^2={\rm e}^{\frac12\gamma t}\left(\tfrac1\gamma\mathrm D_y+\tfrac12\mathrm D_x\right),\quad
\tilde{\mathrm P}^3={\rm e}^{\frac32\gamma t}\left(\tfrac1\gamma\mathrm D_y+\tfrac32\mathrm D_x
+\tfrac32(x-\tfrac12\gamma y)\right),
\\[1.5ex]
\tilde{\mathrm P}^0={\rm e}^{-\frac34\gamma t}\left(\tfrac1\gamma\mathrm D_y-\tfrac34\mathrm D_x\right),\quad
\tilde{\mathrm P}^1={\rm e}^{-\frac14\gamma t}\left(\tfrac2\gamma\mathrm D_y-\tfrac12\mathrm D_x\right),\\
\tilde{\mathrm P}^2={\rm e}^{\frac14\gamma t}\left(\tfrac4\gamma\mathrm D_y+\mathrm D_x
+x+\tfrac34\gamma y\right),\quad
\tilde{\mathrm P}^3={\rm e}^{\frac34\gamma t}\left(\tfrac8\gamma\mathrm D_y+6\mathrm D_x
+6(x+\tfrac14\gamma y)\right).
\end{gather*}
For the equation~\eqref{eq:Power5FP} as an evolution equation,
the choice of canonical representatives of its generalized symmetries
in Corollary~\ref{cor:GenSymsEquivRemarkableFP} is not standard.
One can fix this point via replacing, in view of the equation~\eqref{eq:Power5FP},
the operator~$\mathrm D_t$ in~$\tilde{\mathcal P}^3$,~\dots, $\tilde{\mathcal P}^0$
by the operator~$x^5\mathrm D_x^{\,2}-x\mathrm D_y$.
At the same time, this replacement complicates the expanded form of the representatives and increases their order.

\noprint{
Elements $\tilde{\mathcal P}^0$, \dots $\tilde{\mathcal P}^3$
of the radicals of the corresponding essential Lie invariance algebras
take the form
\begin{gather*}
\tilde{\mathcal P}^0=\p_t,\quad
\tilde{\mathcal P}^1=y\p_t-x^2\p_x-xu\p_u,\\
\tilde{\mathcal P}^2=y^2\p_t-2x^2y\p_x-(2yx+x^{-1})u\p_u,\quad
\tilde{\mathcal P}^3=y^3\p_t-3x^2y^2\p_x-3(xy^2-t+x^{-1}y)u\p_u,
\\[1.5ex]
\tilde{\mathcal P}^0={\rm e}^{-\frac32\gamma t}\left(\tfrac1\gamma\p_y-\tfrac32\p_x\right),\quad
\tilde{\mathcal P}^1={\rm e}^{-\frac12\gamma t}\left(\tfrac1\gamma\p_y-\tfrac12\p_x
+\tfrac12(x+\tfrac32\gamma y)u\p_u\right),\\
\tilde{\mathcal P}^2={\rm e}^{\frac12\gamma t}\left(\tfrac1\gamma\p_y+\tfrac12\p_x\right),\quad
\tilde{\mathcal P}^3={\rm e}^{\frac32\gamma t}\left(\tfrac1\gamma\p_y+\tfrac32\p_x
-\tfrac32(x-\tfrac12\gamma y)u\p_u\right),
\\[1.5ex]
\tilde{\mathcal P}^0={\rm e}^{-\frac34\gamma t}\left(\tfrac1\gamma\p_y-\tfrac34\p_x\right),\quad
\tilde{\mathcal P}^1={\rm e}^{-\frac14\gamma t}\left(\tfrac2\gamma\p_y-\tfrac12\p_x\right),\\
\tilde{\mathcal P}^2={\rm e}^{\frac14\gamma t}\left(\tfrac4\gamma\p_y+\p_x-(x+\tfrac34\gamma y)u\p_u\right),\quad
\tilde{\mathcal P}^3\todo={\rm e}^{\frac34\gamma t}\left(\tfrac8\gamma\p_y+6\p_x-6(x+\tfrac14\gamma y)u\p_u\right).
\end{gather*}
}

\section{Generalized symmetries and finding exact solutions}\label{sec:ExactSolutions}

Generalized symmetries can be used for finding exact solutions of systems of differential equations,
advancing the well-established methods that are based on Lie symmetries.
However, the computations involving generalized symmetries
are much more complicated than their counterparts involving Lie symmetries.
This is why such usage of generalized symmetries is not quite common in the literature.
Furthermore, in the case of linear systems of differential equations,
it has a number of specific features, which we would like to discuss in general
before the particular consideration of the equation~\eqref{eq:RemarkableFP}.

\subsection{Theoretical background}\label{sec:ExactSolutionsTheory}

Given a system~$\mathcal L$ of differential equations
for unknown functions $u=(u^1,\dots,u^m)$ of independent variables $x=(x_1,\dots,x_n)$, $m,n\in\mathbb N$,
with point-symmetry (pseudo)group~$G$, maximal Lie invariance algebra~$\mathfrak g$
and generalized-symmetry algebra~$\Sigma$,
one can use these objects or their parts to construct particular exact solutions of~$\mathcal L$.
The range of possibilities include the following \cite{blum2010A,blum1989A,boch1999A,ibra1985A,kras1986A,olve1977a}.

\medskip\par\noindent
\emph{Generation of solutions via acting by point symmetries}.
If $u=f(x)$ is a solution of~$\mathcal L$ and $\Phi$ is an arbitrary element of~$G$,
then $u=\Phi_*f(x)$ (resp.\ $u=\Phi^*f(x)$) is a solution of~$\mathcal L$ as well.
Here $\Phi_*$ and $\Phi^*$ denote the pushforward and pullback of functions by
the local diffeomorphism~$\Phi$, respectively.
This induce an equivalence relation of the solution set of~$\mathcal L$,
which we call the $G$-equivalence.
All the constructions of exact solutions of~$\mathcal L$ can be carried out up to
this equivalence.

\medskip\par\noindent
\emph{Lie reductions}.
For any subalgebra~$\mathfrak s$ of~$\mathfrak g$ that satisfies the transversality condition
\cite[Eq.~(3.34)]{olve1993A},
the system~$\mathcal S$ of differential constraints $Q_1[u]=0$, \dots, $Q_r[u]=0$,
where $Q_1$,~\dots, $Q_r$ constitute a basis of~$\mathfrak s$
and $Q_i[u]$ denotes the characteristic of~$Q_i$, $i=1,\dots,r$,
is formally compatible with~$\mathcal L$.
The solution set of the extended system $\mathcal L\cup\mathcal S$
coincides with the set of $\mathfrak s$-invariant solutions of~$\mathcal L$.
Since the solution of~$\mathcal S$ reduces
to the solution of a system of $r$ first-order quasilinear partial differential equations
with respect to a single function of $n+m$ independent variables,
the system~$\mathcal S$ can be integrated, which gives an ansatz for $u$
in terms of a tuple~$\phi$ of new unknown functions of $n-r$ arguments.
Substituting this ansatz to~$\mathcal L$, one obtain the so-called reduced system~$\mathcal L/\mathfrak s$
for~$\phi$ with a less number of independent variables than that in the original system~$\mathcal L$.
Any solution of~$\mathcal L/\mathfrak s$ gives, via the ansatz, a solution of~$\mathcal L$.
$G$-equivalent subalgebras of~$\mathfrak g$ result in $G$-equivalent families of invariant solutions.
Hence, in view of the previous points,
only $G$-inequivalent subalgebras of~$\mathfrak g$ of should be used for Lie reductions.

\medskip\par\noindent
\emph{Solutions invariant with respect to generalized symmetries}.
Similarly to Lie reductions, consider a (finite-dimensional) subalgebra~$\Theta$ of~$\Sigma$
with a basis $(\mathfrak Q_1,\dots,\mathfrak Q_r)$.
The system~$\mathcal S$ of differential constraints $\mathfrak Q_1[u]=0$, \dots, $\mathfrak Q_r[u]=0$,
where $\mathfrak Q_i[u]$ denotes the characteristic of~$\mathfrak Q_i$, $i=1,\dots,r$,
is formally compatible with~$\mathcal L$.
At the same time, in contrast to Lie reductions, there is no unified way to integrate
the extended system $\mathcal L\cup\mathcal S$,
and the deep analysis of the structure of the algebra~$\Sigma$ becomes important in this context.
Since the group~$G$ acts on generalized symmetries as well,
only $G$-inequivalent subalgebras of~$\Sigma$ should be used for
constructing solutions that are invariant with respect to generalized symmetries.

\medskip

All the above techniques also work if $\mathcal L$ is a homogeneous linear system of differential equations,
but their application has specific features in this case.
In particular, the generation of solutions via acting by point symmetries
can be extended using the intermediate complexification.
More specifically, if the system~$\mathcal L$ has real-analytical coefficients,
we can complexify it, assuming both independent and dependent variables to be complex.
After finding families of exact solutions of the complexified system~$\mathcal L$,
we extend these families by acting with the complexified version of the group~$G$.
Then assuming the independent variables to be real,
we take the real and imaginary parts of the obtained solutions,
thus constructing new families of (real) exact solutions of~$\mathcal L$.

Moreover, denote by~$\Lambda$ the algebra of linear generalized symmetries of~$\mathcal L$,
which is a subalgebra of~$\Sigma$,
and by~$\Upsilon$ the associative algebra of linear differential recursion operators of~$\mathcal L$.
As discussed in the beginning of Section~\ref{sec:RemarkableFPPolynomialSols},
$\mathrm Q\in\Upsilon$ if and only if $(\mathrm Qu)\p_u\in\Lambda$ \cite[Proposition~5.22]{olve1993A},
and then the function~$\mathrm Qh$ is a solution of~$\mathcal L$ whenever the function~$h$ is.
This gives one more specific way for generating solutions of linear systems of differential equations.
One should select only $G^{\rm ess}$-inequivalent elements%
\footnote{%
It is a quite common situation for a linear system~$\mathcal L$ of differential equations that
the corresponding point-symmetry (pseudo)group~$G$ consists of fiber-preserving transformations
that are affine with respect to the dependent variables.
Let $G^{\rm ess}$ denote the (pseudo)subgroup of all elements of~$G$
that induce homogeneous linear transformations of the unknown functions of~$\mathcal L$.
The natural action of group~$G^{\rm ess}$ preserves
the algebra~$\Lambda$ of linear generalized symmetries of~$\mathcal L$.
Hence the $G^{\rm ess}$-equivalence of these symmetries is well-defined
and induces the $G^{\rm ess}$-equivalence of linear differential recursion operators of~$\mathcal L$.
}
of~$\Upsilon$, but even this selection does not guarantee nontrivial results.
In addition, it is necessary to carefully analyze the action
of particular elements of~$\Upsilon$ on known families of solutions of~$\mathcal L$~\cite{kova2024a}.
For convenience, we call $\langle(\mathrm Qu)\p_u\rangle$-invariant solutions of~$\mathcal L$
merely \emph{$\mathrm Q$-invariant solutions} of~$\mathcal L$.
Thus, a solution~$h$ of~$\mathcal L$ is $\mathrm Q$-invariant if and only if $\mathrm Qh=0$.
It is clear that given an element~$\mathrm Q$ of~$\Upsilon$,
the action by any element from the principal left ideal $\Upsilon\mathrm Q$
on any $\mathrm Q$-invariant solution of~$\mathcal L$ results in the zero solution of~$\mathcal L$.
More generally, any element~$\mathrm R$ of~$\Upsilon$ with $\mathrm Q\mathrm R\in\Upsilon\mathrm Q$
%for which there exists $\mathrm R'\in\Upsilon$ such that
%the operator~$\mathrm Q$ is intertwining for $(\mathrm R,\mathrm R')$, $\mathrm Q\mathrm R=\mathrm R'\mathrm Q$,
maps the space of $\mathrm Q$-invariant solutions of~$\mathcal L$ to itself.
In particular, this is the case for any element of $\Upsilon\mathrm Q+\mathrm C_{\Upsilon}(\mathrm Q)$, i.e.,
for the sum of any elements of the principal left ideal $\Upsilon\mathrm Q$
and of the centralizer~$\mathrm C_{\Upsilon}(\mathrm Q)$ of~$\mathrm Q$ in~$\Upsilon$.
This is why the study of algebraic structure of~$\Upsilon$ is relevant in the context
of generating solutions of~$\mathcal L$ via acting by elements of~$\Upsilon$.

The relation of linear generalized symmetries of linear partial differential equations
with separation of variables in these equations is well known for a long time~\cite{mill1977A}.

For the further consideration, we need some auxiliary statements.

\begin{lemma}\label{lem:OperatorPolymomialKer}
Suppose that $\mathbb F$ is an algebraically closed field,
$V$ is a vector space over~$\mathbb F$, $A\in{\rm End}(V)$,
$P\in\mathbb F[x]$,
$\lambda_1$, \dots, $\lambda_r$ are all the distinct roots of~$P$ with multiplicities~$k_1$, \dots, $k_r$,
respectively.
Denote by ${\rm id}$ the identity endomorphism on~$V$.
Then
\[
\ker P(A)=\bigoplus_{i=1}^r\ker(A-\lambda_i{\rm id})^{k_i}.
\]
\end{lemma}

\begin{proof}
Let $\lambda$ be a root of the polynomial~$P$ with multiplicity~$k$.
Then this polynomial can be factored in the form $P(x)=R(x)(x-\lambda)^k$,
where $R\in\mathbb F[x]$ with $R(\lambda)\ne0$.
For the proof of the lemma, it suffices to show that
\[\ker P(A)=\ker R(A)\oplus\ker(A-\lambda\,{\rm id})^k.\]

Shifting~$x$, $x-\lambda\mapsto x$, or, equivalently, replacing~$A$ by $A+\lambda\,{\rm id}$,
we can assume without loss of generality that $\lambda=0$,
and thus $P(x)=R(x)x^k$ with $R(0)\ne0$.
We represent $R$ as $R(x)=R(0)(1-Q(x))$, where $Q\in x\mathbb F[x]$.
Let $v\in\ker P(A)$. Then $R(A)v=w$, where $w\in\ker A^k$.
Hence $v=\hat v+\check v$, where $\hat v\in\ker R(A)$ and
\[
\check v=\frac1{R(0)}\sum_{j=0}^k\big(Q(A)\big)^jw\in\ker A^k.
\]
Therefore, $\ker P(A)=\ker R(A)+\ker A^k$.
Suppose that $v\in\ker R(A)\cap\ker A^k$.
Since the polynomial~$R$ and the monomial~$x^k$ are coprime,
according to B\'ezout's identity,
there exist polynomials~$N$ and~$M$ such that $N(x)R(x)+M(x)x^k=1$.
Then $v=\big(N(A)R(A)+M(A)A^k\big)v=0$.
This proves that the above sum of kernels is direct.
\end{proof}

\begin{corollary}\label{cor:OperatorPolymomialKer}
Suppose that $V$ is a vector space over a field~$\mathbb F$, $A\in{\rm End}(V)$,
$P\in\mathbb F[x]$, and $P=P_1\cdots P_s$ is a factorization with coprime polynomials~$P_1$, \dots, $P_s$.
Then
\[\ker P(A)=\ker P_1(A)\oplus\dots\oplus\ker P_s(A).\]
\end{corollary}

\begin{proof}
If the field~$\mathbb F$ is algebraically closed, the corollary's statement
directly follows from Lemma~\ref{lem:OperatorPolymomialKer}.

Otherwise, let $\bar{\mathbb F}$ be the algebraic closure of~$\mathbb F$,
and $\bar V:=\bar{\mathbb F}\otimes_{\mathbb F}V$ and
$\bar A:=1_{\bar{\mathbb F}}\otimes_{\mathbb F}A\colon\bar V\to\bar V$
be the counterparts of~$V$ and~$A$ under this closure.
The vector space is naturally embedded in~$\bar V$ via identifying it with $1_{\bar{\mathbb F}}\otimes_{\mathbb F}V$.
In view of Lemma~\ref{lem:OperatorPolymomialKer}, we have
${\ker P(\bar A)=\ker P_1(\bar A)\oplus\dots\oplus\ker P_s(\bar A)}$.
Since $\ker P(A)=\ker P(\bar A)\cap V$ and $\ker P_j(A)=\ker P_j(\bar A)\cap V$, $j=1,\dots,s$,
we obtain the required equality for~$A$.
\end{proof}

\begin{lemma}\label{lem:OperatorPowerKer}
Suppose that $V$ is a vector space over a field~$\mathbb F$,
$A,B\in{\rm End}(V)$, $C:=[A,B]$ commutes with~$A$ and $C\ker A=\ker A$.
Then for any $r\in\mathbb N$,
\begin{gather}\label{eq:OperatorPowerKer}
\ker A^r=\sum_{i=0}^{r-1}B^i\ker A.
\end{gather}
\end{lemma}

The condition $[A,C]=0$ implies $C\ker A\subseteq\ker A$,
but we need the stronger condition $C\ker A=\ker A$.

\begin{proof}
Denote the direct sum in the right-hand side of the last equality by~$U_r$.

We can show by induction with respect to~$i$ that for any~$r>i$,
\[
A^rB^i=\big(rC+BA\big)\big((r-1)C+BA\big)\cdots\big((r-i+1)C+BA\big)A^{r-i}.
\]
Therefore, $B^i\ker A\subseteq\ker A^r$ for any~$r>i$, $i\in\mathbb N$,
which obviously implies $U_r\subseteq\ker A^r$.

The inverse inclusion $\ker A^r\subseteq U_r$ is proved by the induction with respect to~$r$.
The induction base $r=1$ is obvious since $U_1=\ker A$.
Let $r>1$. The induction hypothesis is $\ker A^{r-1}\subseteq U_{r-1}$.
Suppose that $v\in\ker A^r$.
Then $A^{r-1}v=w$, where $w\in\ker A$.
Since $C\ker A=\ker A$, then $C^k\ker A=\ker A$ for any $k\in\mathbb N$.
For any $\hat v\in\ker A$, we have $A^{r-1}B^{r-1}\hat v=(r-1)!C^{r-1}\hat v$.
Therefore, there exists $\hat v\in\ker A$ such that $A^{r-1}B^{r-1}\hat v=w$,
and then $\check v:=v-B^{r-1}\hat v\in\ker A^{r-1}$.
As a result, $v=B^{r-1}\hat v+\check v\in B^{r-1}\ker A+\ker A^{r-1}\subseteq B^{r-1}\ker A+U_{r-1}=U_r$.
\end{proof}

\begin{remark}\label{rem:OperatorPowerKer}
If in addition to the condition of Lemma~\ref{lem:OperatorPowerKer},
we have $\ker A\cap\ker C=\{0\}$,
then the sum in~\eqref{eq:OperatorPowerKer} is direct.
Indeed, then the operator~$C$ and, therefore, all its powers are injective on~$\ker A$.
Let $v\in B^i\ker A\cap B^j\ker A$, $i<j$.
Then $0=A^jv=A^jB^i\hat v=A^jB^j\hat v=(j-1)!C^{j-1}\hat v$,
which implies that $v=0$.
\end{remark}

\begin{corollary}\label{cor:OperatorPowerKer}
Suppose that $V$ is a vector space over a field~$\mathbb F$,
$A,B\in{\rm End}(V)$ such that $[A,B]=c\,{\rm id}$ with some nonzero $c\in\mathbb F$.
Then for any $r\in\mathbb N$,
\[
\ker A^r=\bigoplus_{i=0}^{r-1}B^i\ker A.
\]
\end{corollary}

We use the above assertions in the context of finding
solutions that are invariant with respect to generalized symmetries.
If generalized vector field $(\mathrm Qu)\p_u$ is a (linear) generalized symmetry of~$\mathcal L$,
then for any polynomial~$P$ of~$\mathrm Q$,
the generalized vector field $(P(\mathrm Q)u)\p_u$ is such as well.
If $P=P_1\cdots P_s$ is a factorization with coprime polynomials~$P_1$,~\dots, $P_s$,
then the space of $P(\mathrm Q)$-invariant solutions of~$\mathcal L$
is the direct sum of the spaces of $P_j(\mathrm Q)$-invariant solutions of~$\mathcal L$, $j=1,\dots,s$.
Therefore, when finding $P(\mathrm Q)$-invariant solutions of~$\mathcal L$
for fixed~$\mathrm Q$ and various~$P$,
it suffices to consider only the polynomials $P(x)=(x-\lambda)^k$, $k\in\mathbb N$, $\lambda\in\mathbb C$,
and, for $\lambda\in\mathbb C\setminus\mathbb R$ in the real case,
realify the corresponding space of invariant solutions.
In addition, one should select only $G^{\rm ess}$-inequivalent recursion operators~$\mathrm Q$,
where the action of~$G^{\rm ess}$ on such operators, $\mathrm Q$,
is induced by the action on the corresponding linear generalized symmetries, $(\mathrm Qu)\p_u$.
Moreover, finding $(\mathrm Q-\lambda E)^k$-invariant solutions of~$\mathcal L$,
where $E$ is the $m\times m$ identity matrix,
is simplified when the operator $\mathrm Q-\lambda E$ fits into the framework of Lemma~\ref{lem:OperatorPowerKer}.
More specifically, we assume that $V$ is the solution space of~$\mathcal L$.
If there exists $\mathrm S\in\Upsilon$ such that
the operator $[\mathrm Q,\mathrm S]$ commutes with~$\mathrm Q$
and its restriction to the kernel of~$\mathrm Q-\lambda E$ is a surjection,
then the space of $(\mathrm Q-\lambda E)^k$-invariant solutions of~$\mathcal L$
can be constructed as the sum of the images
of the space of $(\mathrm Q-\lambda E)$-invariant solutions of~$\mathcal L$
under successive action of~$\mathrm S$ up to $k-1$ times.

In fact, the generalized reduction procedure can be systematically realized only
for linear generalized symmetries of the form $(P(\mathrm Q)u)\p_u$,
where $\mathrm Q$ is in addition a first-order Lie-symmetry operator,
$\mathrm Q=(\xi^i(x)\mathrm D_i)E-H(x)$ with
an $m\times m$ matrix function $H(x)=\big(\eta^{ab}(x)\big)$ of~$x$, and $\xi^i\ne0$ for some~$i$.
Here and in what follows the index~$i$ runs from~1 to~$n$,
the indices~$a$ and~$b$ run from~1 to~$m$,
and we assume summation with respect to repeated indices.
As stated above for the case of general~$\mathrm Q$,
it suffices to consider very particular polynomials~$P$.
More specifically, let $\lambda_1$,~\dots, $\lambda_r$ be all the distinct roots of~$P$ over~$\mathbb C$,
and let $k_j$ be the multiplicity of~$\lambda_j$, $j=1,\dots,r$.
Then the space of $P(\mathrm Q)$-invariant solutions of~$\mathcal L$
decomposes into the direct sum of the spaces of
$(\mathrm Q-\lambda_jE)^{k_j}$-invariant solutions of~$\mathcal L$ for each $\lambda_j\in\mathbb R$
and
$(\mathrm Q-\lambda_jE)^{k_j}(\mathrm Q-\bar\lambda_jE)^{k_j}$-invariant solutions of~$\mathcal L$
for each unordered pair $\{\lambda_j,\bar\lambda_j\}$ with $\lambda_j\in\mathbb C\setminus\mathbb R$,
where $\bar\lambda_j$ denotes the complex conjugate of~$\lambda_j$.
In other words, the construction of solutions of~$\mathcal L$ that are invariant with respect to
generalized symmetries associated with polynomials of a single Lie-symmetry operators~$\mathrm Q$
reduces to the case when polynomials are powers of binomials $\mathrm Q-\lambda E$
and, if $\lambda\in\mathbb C\setminus\mathbb R$,
the separation of the real and imaginary parts of the obtained solutions.
For each binomial power $(\mathrm Q-\lambda E)^k$,
the corresponding ansatz is just a representation of the general solution of the system $(\mathrm Q-\lambda E)^ku=0$,
and thus it takes the form
\begin{gather}\label{eq:GenRedAnsatzGeneralForm}
u^a={\rm e}^{\lambda\zeta}\sum_{j=0}^{k-1}\sum_{b=1}^mf^{ab}(x)\varphi^{bj}(\omega_1,\dots,\omega_{n-1})\zeta^j,
\quad a=1,\dots,m.
\end{gather}
Here
$\omega_1=\omega_1(x)$, \dots, $\omega_{n-1}=\omega_{n-1}(x)$ are functionally independent solutions
of the equation $\xi^i(x)\p_i\omega=0$,
$\zeta=\zeta(x)$ is a particular solution of the equation $\xi^i(x)\p_i\zeta=1$, and
$(f^{1b},\dots,f^{mb})$ are linearly independent solutions of the system $\xi^i(x)\p_if^a=\eta^{ab}(x)f^b$.
The selection of $G^{\rm ess}$-inequivalent Lie-symmetry operators reduces to
the classification of one-dimensional subalgebras
in the essential Lie invariance algebra of the system~$\mathcal L$.
Substituting the ansatz~\eqref{eq:GenRedAnsatzGeneralForm} into this system
results in a homogeneous linear system of differential equations
with respect to the new unknown functions~$\varphi^{bj}$
of less number of independent variables~$\omega_1$,~\dots,~$\omega_{n-1}$.

\subsection{Point symmetries and Lie-invariant solutions}

Now we present the results on point symmetries and Lie reductions of the equation~\eqref{eq:RemarkableFP}
that were obtained in~\cite{kova2023a,kova2024c}.

\begin{theorem}[\cite{kova2023a}]\label{thm:RemarkableFPSymGroup}
The complete point symmetry pseudogroup~$G$ of the remarkable Fokker--Planck equation~\eqref{eq:RemarkableFP}
consists of the transformations of the form
\begin{gather}\label{eq:RemarkableFPSymGroup}
\begin{split}
&\tilde t=\frac{\alpha t+\beta}{\gamma t+\delta},
\quad
\tilde x=\frac{\hat x}{\gamma t+\delta}
-\frac{3\gamma\hat y}{(\gamma t+\delta)^2},
\quad
\tilde y=\frac{\hat y}{(\gamma t+\delta)^3},
\\[1ex]
&\tilde u=\sigma(\gamma t+\delta)^2\exp\left(
\frac{\gamma\hat x^2}{\gamma t+\delta}
-\frac{3\gamma^2\hat x\hat y}{(\gamma t+\delta)^2}
+\frac{3\gamma^3\hat y^2}{(\gamma t+\delta)^3}
\right)
\\
&\hphantom{\tilde u={}}
\times\exp\big(
3\lambda_3(y-tx)-\lambda_2x-(3\lambda_3^2t^3+3\lambda_3\lambda_2t^2+\lambda_2^2t)
\big)
\big(u+f(t,x,y)\big),
\end{split}
\end{gather}
where
$\hat x:=x+3\lambda_3t^2+2\lambda_2t+\lambda_1$,
$\hat y:=y+\lambda_3t^3+\lambda_2t^2+\lambda_1t+\lambda_0$;
$\alpha$, $\beta$, $\gamma$ and $\delta$ are arbitrary constants with $\alpha\delta-\beta\gamma=1$;
$\lambda_0$,~\dots, $\lambda_3$ and $\sigma$ are arbitrary constants with $\sigma\ne0$,
and $f$ is an arbitrary solution of~\eqref{eq:RemarkableFP}.	
\end{theorem}

Pulling back an arbitrary solution $u=h(t,x,y)$ of~\eqref{eq:RemarkableFP}
by an arbitrary point symmetry transformation of the form~\eqref{eq:RemarkableFPSymGroup},
we obtain, in the notation of Theorem~\ref{thm:RemarkableFPSymGroup},
the formula of generating new solutions of~\eqref{eq:RemarkableFP}
from known ones under the action of elements of~$G$,
\begin{gather}\label{eq:RemarkableFPNewSolutionsByG0}
\begin{split}
u={}&
\frac{{\rm e}^{\lambda_2x-3\lambda_3(y-tx)+3\lambda_3^2t^3+3\lambda_3\lambda_2t^2+\lambda_2^2t}}{\sigma(\gamma t+\delta)^2}
\exp\left(
-\frac{\gamma\hat x^2}{\gamma t+\delta}
+\frac{3\gamma^2\hat x\hat y}{(\gamma t+\delta)^2}
-\frac{3\gamma^3\hat y^2}{(\gamma t+\delta)^3}
\right)
\\
&
\times
h\left(
\frac{\alpha t+\beta}{\gamma t+\delta},\,
\frac{\hat x}{\gamma t+\delta}-\frac{3\gamma\hat y}{(\gamma t+\delta)^2},\,
\frac{\hat y}{(\gamma t+\delta)^3}
\right)
-f(t,x,y).
\end{split}
\end{gather}
Due to the complexification trick, applying the formula~\eqref{eq:RemarkableFPNewSolutionsByG0}
for generating solutions to a real analytical solution of the equation~\eqref{eq:RemarkableFP},
one can assume all the constant parameters in~\eqref{eq:RemarkableFP} to be complex
and then take the real and imaginary parts of the obtained solutions.

We use the modified transformation composition as the group operation in $G$.
More specifically, this composition respects the natural domains of transformations of the form~\eqref{eq:RemarkableFPNewSolutionsByG0},
see~\cite[Section~3]{kova2023a} for details.
The point transformations of the form
\[
\mathscr Z(f)\colon\quad \tilde t=t,\quad \tilde x=x,\quad \tilde y=y,\quad \tilde u=u+f(t,x,y),
\]
where the parameter function $f=f(t,x,y)$ is an arbitrary solution of the equation~\eqref{eq:RemarkableFP},
are associated with the linear superposition of solutions of this equation
and thus can be considered as trivial.
They constitute the normal pseudosubgroup $G^{\rm lin}$ of the pseudogroup $G$.
The pseudogroup~$G$ splits over~$G^{\rm lin}$, $G=G^{\rm ess}\ltimes G^{\rm lin}$,
where $G^{\rm ess}$ is the \emph{subgroup} of~$G$ consisting of the transformations of the form~\eqref{eq:RemarkableFPSymGroup} with $f=0$
and with their natural domains, and thus it is an eight-dimensional Lie group.

We exhaustively carried out Lie reductions of the equation~\eqref{eq:RemarkableFP}
in \cite{kova2023a}, beginning with the classification
of $G^{\rm ess}$-inequivalent one- and two-dimensional subalgebras of $\mathfrak g^{\rm ess}$.

\begin{lemma}\label{lem:RemarkableFP1DSubalgs}
A complete list of $G^{\rm ess}$-inequivalent one-dimensional subalgebras of $\mathfrak g^{\rm ess}$ is exhausted by the subalgebras
\begin{gather*}
\mathfrak s_{1.1}=\langle\mathcal P^t+\mathcal P^3\rangle,\ \
\mathfrak s_{1.2}^\delta=\langle\mathcal P^t+\delta\mathcal I\rangle,\ \
\mathfrak s_{1.3}^\nu=\langle\mathcal D+\nu\mathcal I\rangle,\ \
\mathfrak s_{1.4}^\mu=\langle\mathcal P^t+\mathcal K+\mu\mathcal I\rangle,\\
\mathfrak s_{1.5}^\varepsilon=\langle\mathcal P^2+\varepsilon\mathcal P^0\rangle,\ \
\mathfrak s_{1.6}=\langle\mathcal P^1\rangle,\ \
\mathfrak s_{1.7}=\langle\mathcal P^0\rangle,\ \
\mathfrak s_{1.8}=\langle\mathcal I\rangle,
\end{gather*}
where $\varepsilon\in\{-1,1\}$, $\delta\in\{-1,0,1\}$, and $\mu$ and~$\nu$ are arbitrary real constants with $\nu\geqslant0$.
\end{lemma}
The families of solutions of the equation~\eqref{eq:RemarkableFP}
that are invariant with respect to the subalgebras
$\mathfrak s_{1.2}^0$, $\mathfrak s_{1.5}^\varepsilon$, $\mathfrak s_{1.6}$ and $\mathfrak s_{1.7}$
are parameterized by the general solutions of the (1+1)-dimensional linear heat equations
with the zero and the inverse square potentials,
\begin{gather}\label{eq:F0HeatSolution0}
\solution
u=|x|^{-\frac14}\vartheta^\mu\Big(\tfrac94\tilde\varepsilon y,|x|^{\frac32}\Big)
\quad\mbox{with}\quad \mu=\tfrac5{36},\quad \tilde\varepsilon:=\sgn x,
\\[.5ex]\label{eq:F0HeatSolution1}
\solution
u=|t|^{-\frac12}{\rm e}^{-\frac{x^2}{4t}}\vartheta^0
\Big(\tfrac13{t^3}+2\varepsilon t-t^{-1},2y-(t+\varepsilon t^{-1})x\Big)
\quad\mbox{with}\quad \varepsilon\in\{-1,1\},
\\[.5ex]\label{eq:F0HeatSolution2}
\solution
u=\vartheta^0\Big(\tfrac13t^3,y-tx\Big),
\\[.5ex]\label{eq:F0HeatSolution3}
\solution
u=\vartheta^0(t,x),
\end{gather}
where $\vartheta^\mu=\vartheta^\mu(z_1,z_2)$
is an arbitrary solution of the equation $\vartheta^\mu_1=\vartheta^\mu_{22}+\mu z_2^{-2}\vartheta^\mu$.
It was shown in~\cite{kova2023a} that only Lie reductions of codimension one give essentially new solutions.
The solutions obtained using Lie reductions of codimensions two and three are superfluous
since they give solutions that are $G^{\rm ess}$-equivalent to
elements of the families~\eqref{eq:F0HeatSolution0}--\eqref{eq:F0HeatSolution3}
corresponding to known invariant solutions of the equations $\vartheta^\mu_1=\vartheta^\mu_{22}+\mu z_2^{-2}\vartheta^\mu$ with $\mu\in\{0,\frac5{36}\}$.
Recall that a complete collection of inequivalent Lie invariant solutions
of the (1+1)-dimensional linear heat equation ($\mu=0$)
was presented in Examples 3.3 and 3.17 in~\cite{olve1993A}
and then enhanced in~\cite[Section~A]{vane2021a}.
An analogous collection for all nonzero values of $\mu$ was constructed in \cite[Section~A]{kova2023a},
see also~\cite{gung2018a,gung2018b}.

The sole possibility to extend the Lie reduction procedure using the complexification trick
is to assume the real constants appearing in optimal lists of subalgebras,
like~$\mu$ and~$\nu$ in Lemma~\ref{lem:RemarkableFP1DSubalgs}, to be complex.

\subsection{Generation of solutions by recursion operators}\label{sec:ExactSolutionsGenerationFP}

The analysis of generating solutions of the remarkable Fokker--Planck equation~\eqref{eq:RemarkableFP}
via acting by elements of the associative algebra~$\Upsilon$ of its linear differential recursion operators
was initiated in~\cite{kova2023a,kova2024c}.
Therein, only solutions that are invariant with respect to one-dimensional subalgebras
of the (nil)radical~$\mathfrak r$ of~$\mathfrak g^{\rm ess}$ were considered as seed ones.
In this section, we systematically analyze the use of arbitrary Lie-invariant solutions of~\eqref{eq:RemarkableFP}
within this framework.
It suffices to consider the solutions of~\eqref{eq:RemarkableFP}
that are invariant with respect to the $G^{\rm ess}$-inequivalent one-dimensional subalgebras
of~$\mathfrak g^{\rm ess}$ that are listed in Lemma~\ref{lem:RemarkableFP1DSubalgs}.
Recall that the subalgebra~$\mathfrak s_{1.8}$ is not appropriate for Lie reduction.
Acting on an arbitrary solution $u=h(t,x,y)$ of the equation~\eqref{eq:RemarkableFP} by an element
\begin{gather}\label{eq:UpsilonGenElement}
\mathrm Q=\sum_{(i_3,i_2,i_1,i_0)\in\mathbb N_0^4}c_{i_3i_2i_1i_0}
(\mathrm P^3)^{i_3}(\mathrm P^2)^{i_2}(\mathrm P^1)^{i_1}(\mathrm P^0)^{i_0}
\end{gather}
of the algebra~$\Upsilon=\Upsilon_{\mathfrak r}$,
where only finitely many real constants $c_{i_3i_2i_1i_0}$ are nonzero,
we obtain a~solution $\mathrm Qh$ of~\eqref{eq:RemarkableFP}.
The problem is that applying this procedure to a known solution, one may obtain a solution that is known as well.

For each of fixed subalgebra~$\mathfrak s_{1.k}^*$ from Lemma~\ref{lem:RemarkableFP1DSubalgs},
where $*$ denotes a value of the tuple of parameters of the $k$th subalgebra family, $k=1,\dots,7$,
we denote by~$\mathscr S_k^*$ the family of $\mathfrak s_{1.k}^*$-invariant solutions of~\eqref{eq:RemarkableFP}
and by $\mathrm B$ the element of~$\Upsilon$
that is associated up to the multiplier $-1$ with its canonical basis element.
We directly compute low-degree generators%
\footnote{%
Finding the entire centralizers~$\mathrm C_\Upsilon(\mathrm B)$
is a nontrivial problem for most elements~$\mathrm B$ of~$\Upsilon\simeq{\rm W}(2,\mathbb R)$.
See~\cite{dixm1968a} and references therein,
where the analogous problem is considered for the simpler first Weyl algebra ${\rm W}(1,\mathbb R)$.
}
of the centralizer~$\mathrm C_\Upsilon(\mathrm B)$ of~$\mathrm B$ in~$\Upsilon$
and expand an arbitrary $\mathrm Q\in\Upsilon$, i.e., an arbitrary operator of the form~\eqref{eq:UpsilonGenElement}
with respect to a basis of~$\Upsilon$
in which powers of these generators are right multipliers in the basis elements.
Since the action by elements of~$\mathrm C_\Upsilon(\mathrm B)$ preserves the family~$\mathscr S_k^*$,
such an expansion allows us to determine the action of which elements of~$\Upsilon$
may lead, up to linearly combining solutions, to new solutions of~\eqref{eq:RemarkableFP}.
In the course of this analysis, we have to use the explicit representation~\eqref{eq:PDKRepresentations}
of the operators~$\hat{\mathrm P}^t$, $\hat{\mathrm D}$ and~$\hat{\mathrm K}$
in terms of the canonical generators of the algebra~$\Upsilon$.

In the expansions of~$\mathrm Q$ below,
the indices~$j$ run through~$\mathbb N_0$,
we assume summation with respect to repeated indices,
and only a finite number of operator polynomials~$R$ indexed by tuples of indices~$j$ are nonzero.

\medskip\par\noindent
$\boldsymbol{\mathfrak s_{1.1}.}$
The centralizer $\mathrm C_\Upsilon(\mathrm B)$ of $\mathrm B:=\hat{\mathrm P}^t+\mathrm P^3$ in the algebra $\Upsilon$
contains the operators
\[
\mathrm H^1:=\mathrm P^1-\frac16(\mathrm P^0)^2,\quad
\mathrm H^2:=\mathrm P^2+\frac2{27}(\mathrm P^0)^3-\frac23\mathrm P^1\mathrm P^0,\quad
\mathrm B   =\mathrm P^3-\mathrm P^2\mathrm P^0+(\mathrm P^1)^2.
\]
We expand the operator~$\mathrm Q$ as
$\mathrm Q=R^{j_1j_2 j_3}(\mathrm P^0)\,(\mathrm H^1)^{j_1}(\mathrm H^2)^{j_2}\mathrm B^{j_3}.$
It implies that action of an arbitrary element $\mathrm Q$ of~$\Upsilon$ on
an arbitrary solution $h$ from the family~$\mathscr S_1$ is a (finite) linear combinations
of solutions of the form $(\mathrm P^0)^j\hat h$, where $j\in\mathbb N_0$ and $\hat h\in\mathscr S_2^0$.
Thus, when using seed solutions from~$\mathscr S_1$,
only acting with powers of $\mathrm P^0$ may give solutions out of~$\mathscr S_1$.

\medskip\par\noindent
$\boldsymbol{\mathfrak s_{1.2}^\delta.}$
The operator $\mathrm B:=\hat{\mathrm P}^t+\delta$ commutes with the operators
\[
\mathrm P^0,\quad
\hat{\mathrm P}^t=-\mathrm P^2\mathrm P^0+(\mathrm P^1)^2,\quad
\mathrm H:=\mathrm P^3(\mathrm P^0)^2-3\mathrm P^2\mathrm P^1\mathrm P^0+2(\mathrm P^1)^3.
\]
In view of Lemma~7 from~\cite{kova2023b}, any solution $h\in\mathscr S_2^\delta$
can be represented as $(\mathrm P^0)^k\tilde h$
with $k:=\max\big(\{2i_3+i_2-i_0\mid c_{i_3i_2i_1i_0}\ne0\}\cup\{0\}\big)$ for some $\tilde h\in\mathscr S_2^\delta$,
and thus
\[
\mathrm Qh=\mathrm Q(\mathrm P^0)^k\tilde h
=R^{j_1j_2j_3}(\mathrm P^1)\,(\mathrm P^0)^{j_1}(\hat{\mathrm P}^t)^{j_2}\mathrm H^{j_3}\tilde h.
\]
It implies that action of an arbitrary element $\mathrm Q$ of~$\Upsilon$ on
an arbitrary solution $h$ from the family~$\mathscr S_2^\delta$ is a (finite) linear combinations
of solutions of the form $(\mathrm P^1)^j\hat h$, where $j\in\mathbb N_0$ and~$\hat h\in\mathscr S_2^\delta$.
In~other words, the only way for generating new solutions of the equation~\eqref{eq:RemarkableFP}
from its solutions from the family~$\mathscr S_2^0$ via acting by its linear differential recursion operators
is to use powers of~$\mathrm P^1$.

\medskip\par\noindent
$\boldsymbol{\mathfrak s_{1.3}^\nu.}$
Some families of Lie-invariant solutions are preserved when acting by all linear differential recursion operators.
Consider the span~$\mathscr S_3$ of all $\mathfrak s_{1.3}^\nu$-invariant solutions,
where the parameter~$\nu$ runs through~$\mathbb R$, $\mathscr S_3:=\sum_{\nu\in\mathbb R}\mathscr S_3^\nu$.
Since
$\mathrm P^3\hat{\mathrm D}=(\hat{\mathrm D}-3)\mathrm P^3$,
${\mathrm P^2\hat{\mathrm D}=(\hat{\mathrm D}-1)\mathrm P^2}$,
$\mathrm P^1\hat{\mathrm D}=(\hat{\mathrm D}+1)\mathrm P^1$,
$\mathrm P^0\hat{\mathrm D}=(\hat{\mathrm D}+3)\mathrm P^0$ and thus
\[
(\mathrm P^3)^{i_3}(\mathrm P^2)^{i_2}(\mathrm P^1)^{i_1}(\mathrm P^0)^{i_0}(\hat{\mathrm D}-\nu)
=
(\hat{\mathrm D}-\nu-3i_3-i_2+i_1+3i_0)(\mathrm P^3)^{i_3}(\mathrm P^2)^{i_2}(\mathrm P^1)^{i_1}(\mathrm P^0)^{i_0},
\]
the action by $(\mathrm P^3)^{i_3}(\mathrm P^2)^{i_2}(\mathrm P^1)^{i_1}(\mathrm P^0)^{i_0}$
maps $\mathscr S_3^\nu$ into $\mathscr S_3^{\nu'}$ with $\nu'=\nu+3i_3+i_2-i_1-3i_0$.
Therefore, the action of any element~$\mathrm Q$ of the algebra~$\Upsilon$
maps~$\mathscr S_3$ into itself.

\medskip\par\noindent
$\boldsymbol{\mathfrak s_{1.4}^\mu.}$
The centralizer $\mathrm C_\Upsilon(\mathrm B)$ with $\mathrm B=\hat{\mathrm P}^t+\hat{\mathrm K}+\mu$
contains the associative algebra generated by
\begin{gather*}
\hat{\mathrm P}^t+\hat{\mathrm K}=(\mathrm P^2)^2-\mathrm P^3\mathrm P^1-\mathrm P^2\mathrm P^0+(\mathrm P^1)^2,\quad
\mathrm H:=(\mathrm P^3)^2+3\mathrm P^3\mathrm P^1+3\mathrm P^2\mathrm P^0+(\mathrm P^0)^2,\quad \mathrm C,
\\[1ex]
\mathrm S:=
(\mathrm P^3)^3\mathrm P^0-3(\mathrm P^3)^2\mathrm P^2\mathrm P^1+2\mathrm P^3(\mathrm P^2)^3
+3\mathrm P^3(\mathrm P^2)^2\mathrm P^0-6\mathrm P^3\mathrm P^2(\mathrm P^1)^2-\mathrm P^3(\mathrm P^1)^2\mathrm P^0
\\\qquad{}
-\mathrm P^3(\mathrm P^0)^3
+3(\mathrm P^2)^3\mathrm P^1+6(\mathrm P^2)^2\mathrm P^1\mathrm P^0-3\mathrm P^2(\mathrm P^1)^3
+3\mathrm P^2\mathrm P^1(\mathrm P^0)^2
-2(\mathrm P^1)^3\mathrm P^0
\\\qquad{}
-4(\mathrm P^3)^2
+8(\mathrm P^0)^2
+12\mathrm P^2\mathrm P^0
.
\end{gather*}
Applying arguments similar to those above to the pair~$\hat{\mathrm P}^t+\hat{\mathrm K}$ and~$\mathrm H$,
we conclude that only the action by polynomials
\[
\mathrm P^3\mathrm P^2(\mathrm P^1)^{j_1}(\mathrm P^0)^{j_2},\quad
\mathrm P^3(\mathrm P^1)^{j_1}(\mathrm P^0)^{j_2},\quad
\mathrm P^2(\mathrm P^1)^{j_1}(\mathrm P^0)^{j_2},\quad
(\mathrm P^1)^{j_1}(\mathrm P^0)^{j_2}
\]
might result in nontrivial solution generations.
It is not clear how this consideration can be modified using
higher-degree elements like~$\mathrm C$ and~$\mathrm S$.

\medskip\par\noindent
$\boldsymbol{\mathfrak s_{1.5}^\varepsilon,\mathfrak s_{1.6},\mathfrak s_{1.7}.}$
The centralizers~$\mathrm C_\Upsilon(\mathrm B)$ of
$\mathrm B:=\mathrm P^2+\varepsilon\mathrm P^0$, $\mathrm B:=\mathrm P^1$ and $\mathrm B:=\mathrm P^0$
in the algebra $\Upsilon$
contain the associative algebras generated by the sets
$\{\mathrm P^3-3\varepsilon\mathrm P^1,\mathrm P^2,\mathrm P^0\}$,
$\{\mathrm P^3,\mathrm P^1,\mathrm P^0\}$ and
$\{\mathrm P^2,\mathrm P^1,\mathrm P^0\}$, respectively.
For each of these associative algebras,
we expand an arbitrary element $\mathrm Q$ of~$\Upsilon$, involving its elements as described above.
This implies that the only way to construct essentially new solutions
via acting by linear differential recursion operators
starting with the seeds from the families~$\mathscr S_5^\varepsilon$, $\mathscr S_6$ and~$\mathscr S_7$
is to use powers of~$\mathrm P^1$, $\mathrm P^2$ and $\mathrm P^3$, respectively.

\medskip\par
We combine the above results with results of Sections~\ref{sec:ExactSolutionsGenerationFP}
to construct wide families of exact solutions of the equation~\eqref{eq:RemarkableFP}.
More specifically, we have shown
that for the families~\eqref{eq:F0HeatSolution0}, \eqref{eq:F0HeatSolution1}, \eqref{eq:F0HeatSolution2}
and~\eqref{eq:F0HeatSolution3} of solutions of the equation~\eqref{eq:RemarkableFP}
that are invariant with respect to the subalgebras
$\mathfrak s_{1.2}^0$, $\mathfrak s_{1.5}^\varepsilon$, $\mathfrak s_{1.6}$ and $\mathfrak s_{1.7}$,
only the action by the monomials $(\mathrm P^1)^k$, $(\mathrm P^1)^k$, $(\mathrm P^2)^k$ and $(\mathrm P^3)^k$,
$k\in\mathbb N$, respectively, in general leads to essentially new solutions of~\eqref{eq:RemarkableFP},
\begin{gather}\label{eq:F0HeatSolutionGen0}
\solution
u=(\mathrm P^1)^k\Big(|x|^{-\frac14}\vartheta^\mu\big(\tfrac94\tilde\varepsilon y,|x|^{\frac32}\big)\Big)
\quad\mbox{with}\quad \mu=\tfrac5{36},\quad \tilde\varepsilon:=\sgn x,
\\\label{eq:F0HeatSolutionGen1}
\solution u=(\mathrm P^1)^k\Big(|t|^{-\frac12}{\rm e}^{-\frac{x^2}{4t}}
\vartheta^0\big(\tfrac13{t^3}+2\varepsilon t-t^{-1},2y-(t+\varepsilon t^{-1})x\big)
\Big)
\quad\mbox{with}\quad \varepsilon\in\{-1,1\},
\\\label{eq:F0HeatSolutionGen2}
\solution u=(\mathrm P^2)^k\vartheta^0\Big(\tfrac13t^3,y-tx\Big),
\\\label{eq:F0HeatSolutionGen3}
\solution u=(\mathrm P^3)^k\vartheta^0(t,x).
\end{gather}

Even if the action of an element~$\mathrm Q$ of~$\Upsilon$ preserve a solution family,
it may result in an interesting solution generation within this family.
An example of such a generation is presented in Section~\ref{sec:RemarkableFPPolynomialSols}.
The solution $u=1$ is invariant with respect to the four-dimensional subalgebra
$\langle\mathcal P^t,\mathcal D+2\mathcal I,\mathcal P^1,\mathcal P^0\rangle$
of the essential Lie invariance algebra~$\mathfrak g^{\rm ess}$ of~\eqref{eq:RemarkableFP}.
In other words, this solution belongs to $\mathscr S_2^0\cap\mathscr S_2^2\cap\mathscr S_6\cap\mathscr S_7$.
Lemma~\ref{lem:RemarkableFPPolynomialSols} states that
any solution of the equation~\eqref{eq:RemarkableFP}
that is polynomial with respect to~$x$
is polynomial with respect to the entire tuple of independent variables $(t,x,y)$
and is a (finite) linear combination
of the basis polynomials $(\mathrm P^3)^k(\mathrm P^2)^l1$, $k,l\in\mathbb N_0$,
thus belonging to the space~$\mathscr S_3:=\sum_{\nu\in\mathbb R}\mathscr S_3^\nu$.

\subsection{Generalized reductions}

Using the theory developed in Section~\ref{sec:ExactSolutionsTheory},
we also revisit and extend the results of~\cite{kova2023a,kova2024c}
on particular generalized reductions of the equation~\eqref{eq:RemarkableFP}.
According to this theory,
only polynomials of Lie-symmetry operators of the equation~\eqref{eq:RemarkableFP} can be systematically used
in the procedure of its generalized reduction.
In view of Lemma~\ref{lem:OperatorPolymomialKer},
it suffices to just consider the powers of binomials of the form $\mathrm Q-\lambda$,
where, up to the $G^{\rm ess}$-equivalence,
$\mathrm Q$ runs through the Lie-symmetry operators of~\eqref{eq:RemarkableFP} associated with the basis elements
of the subalgebras~$\mathfrak s_{1.k}^*$, $k=1,\dots,7$, from Lemma~\ref{lem:RemarkableFP1DSubalgs}.
Nevertheless, only some families even of these specific invariant solutions
can be described completely.\looseness-1

It turns out that the easiest and the most complete description is achieved
when the corresponding subalgebra is contained in the (nil)radical $\mathfrak r$
of the algebra $\mathfrak g^{\rm ess}$ and thus, modulo the $G^{\rm ess}$-equivalence,
it is one of the subalgebras $\mathfrak s_{1.5}^\varepsilon$, $\mathfrak s_{1.6}$ and $\mathfrak s_{1.7}$
or, equivalently, $\mathrm Q\in\{\mathrm P^2+\varepsilon\mathrm P^0,\mathrm P^1,\mathrm P^0\}$.
In this case, up to the $G^{\rm ess}$-equivalence,
the parameter $\lambda$ in $\mathrm Q-\lambda$ can be set to zero,
i.e., it suffices to merely describe $\mathrm Q^n$-invariant solutions.
In view of the commutation relations
$[\mathrm P^2+\varepsilon\mathrm P^0,\mathrm P^1]=1$,
$[\mathrm P^1,\mathrm P^2]=1$ and
$[\mathrm P^0,\mathrm P^3]=3$,
Corollary~\ref{cor:OperatorPowerKer} implies that the spaces of
$(\mathrm P^2+\varepsilon\mathrm P^0)^n$-, $(\mathrm P^1)^n$- and $(\mathrm P^0)^n$-invariant solutions,
$n\in\mathbb N$, are the direct sums of spaces of solutions of the form~\eqref{eq:F0HeatSolutionGen1},
\eqref{eq:F0HeatSolutionGen2} and~\eqref{eq:F0HeatSolutionGen3} with fixed~$k$, respectively,
where $k$ runs from~0 to~$n-1$.
Combining the complexification trick, the $G^{\rm ess}$-action and the linear superposition of solutions,
we obtain the entire span of solutions that are invariant
with respect to polynomials of Lie-symmetry operators associated with elements of~$\mathfrak r$.

Considering the subalgebra $\mathfrak s_{1.1}$, where $\mathrm Q=\hat{\mathrm P}^t+\mathrm P^3$,
we can also set $\lambda=0$  modulo the $G^{\rm ess}$-equivalence
and, since $[\hat{\mathrm P}^t+\mathrm P^3,\mathrm P^0]=-3$, apply Corollary~\ref{cor:OperatorPowerKer}.
As a result, we conclude that the space
of $(\hat{\mathrm P}^t+\mathrm P^3)^n$-invariant solutions of~\eqref{eq:RemarkableFP}, $n\in\mathbb N$,
is the direct sums of spaces of solutions of the form
$(\mathrm P^0)^kh$ with fixed~$k$ and an arbitrary $h\in\mathscr S_1$, where $k$ runs from~0 to~$n-1$.
These solutions arise as the result of the only essential nontrivial solution generation
using linear differential recursion operators and seed solutions from the set~$\mathscr S_1$.
Recall \cite[Section~5]{kova2023a} that the function $h=h(t,x,y)$ belongs to~$\mathscr S_1$
if and only if
\[
h={\rm e}^{\frac3{10}t(t^4-5tx+10y)}w(z_1,z_2),\quad z_1:=y-\tfrac14t^4,\quad z_1:=x-t^3,
\]
where $w$ is an arbitrary solution of the reduced equation $z_2w_1=w_{22}-3z_1w$.
The problem is that no nonzero solutions of the latter equation
and, therefore, no nonzero elements of~$\mathscr S_1$ are known.

The description of generalized reductions associated with polynomials of~$\hat{\mathrm P}^t$
is more involved comparing to the previous cases.
Up to the $G^{\rm ess}$-equivalence, the parameter~$\lambda$ in the operator $\hat{\mathrm P}^t+\lambda$
can be gauged at most to $\delta\in\{-1,0,1\}$.
In other words, we should consider the operator $\hat{\mathrm P}^t+\delta$
associated with the basis element of the subalgebra $\mathfrak s_{1.2}^\delta$
for each $\delta\in\{-1,0,1\}$.
We have $[\hat{\mathrm P}^t+\delta,\mathrm P^1]=\mathrm P^0$, $[\hat{\mathrm P}^t+\delta,\mathrm P^0]=0$
and, in view of~\cite[Lemma~7]{kova2023b}, $\mathrm P^0\ker(\hat{\mathrm P}^t+\delta)=\ker(\hat{\mathrm P}^t+\delta)$.
Lemma~\ref{lem:OperatorPowerKer} in this setting implies
that the space of $(\hat{\mathrm P}^t+\delta)^n$-invariant solutions of~\eqref{eq:RemarkableFP}, $n\in\mathbb N$,
is the sum of the spaces of solutions of the form
$(\mathrm P^0)^kh$ with fixed~$k$ and an arbitrary $h\in\mathscr S_2^\delta$,
where $k$ runs from~0 to~$n-1$.
The $\mathfrak s_{1.2}^\delta$-invariant solutions are of the form
$u={\rm e}^{\delta t}w(x,y)$, where $w$ is an arbitrary solution of the equation $xw_y=w_{xx}-\delta w$.
For $\delta=0$, the last equation is reduced by a point transformation
to a (1+1)-dimensional linear heat equation with an inverse square potential,
and thus the generated solutions are precisely given by~\eqref{eq:F0HeatSolutionGen0}.
In the case $\delta\ne0$, we were able to construct
only those among the $\mathfrak s_{1.2}^\delta$-invariant solutions that
are in addition $\mathrm P^0$-invariant, i.e., that are $\mathfrak s_{2.2}^\delta$-invariant,
where $\mathfrak s_{2.2}^\delta=\langle\mathcal P^t+\delta\mathcal I,\mathcal P^0\rangle$.
The solution generation using linear differential recursion operators
and seed solutions from the set~$\mathscr S_2^\delta\cap\mathscr S_7$
results only in certain solutions of the form~\eqref{eq:F0HeatSolutionGen3}.
Hence it is necessary to find other $\mathfrak s_{1.2}^\delta$-invariant solutions with $\delta\ne0$,
which is an open nontrivial problem.

The analysis of generalized reductions associated with polynomials of~$\hat{\mathrm D}$
and of~$\hat{\mathrm P}^t+\hat{\mathrm K}$ is far more complicated.
Modulo the $G^{\rm ess}$-equivalence, the parameter~$\lambda$
can at most be made nonnegative in the operator $\hat{\mathrm D}-\lambda$
and cannot be changed in the operator $\hat{\mathrm P}^t+\hat{\mathrm K}-\lambda$.
This is why the consideration of all the subalgebras~$\mathfrak s_{1.3}^\nu$ and $\mathfrak s_{1.4}^\mu$
from Lemma~\ref{lem:RemarkableFP1DSubalgs} is relevant here.
The space~$\mathscr S_3:=\sum_{\nu\in\mathbb R}\mathscr S_3^\nu$
is preserved by the action of any element~$\mathrm Q$ of~$\Upsilon$,
see Section~\ref{sec:ExactSolutionsGenerationFP}.
We can show analogously that there is no obvious way to derive \smash{$(\hat{\mathrm D}-\lambda)^n$}-invariant solutions
from analogous solutions with lower values of~$n$.
The same claim holds for $(\hat{\mathrm P}^t+\hat{\mathrm K}-\lambda)^n$-invariant solutions.
Following the consideration in the last paragraph of Section~\ref{sec:ExactSolutionsTheory},
we construct generalized ansatzes for such solutions and the corresponding reduced systems, which are
\begin{gather*}
u=|t|^{\frac12\lambda-1}\sum_{j=0}^{n-1}w^j(z_1,z_2)\zeta^j,\quad
\omega_1:=|t|^{-\frac32}y,\quad
\omega_2:=|t|^{-\frac12}x,\quad
\zeta:=\frac12\ln|t|,
\\
(2\varepsilon'z_2-3z_1)w^j_1=2\varepsilon'w^j_{22}+z_2w^j_2-(\lambda-2)w^j-(j+1)w^{j+1},\quad
j=0,\dots,n-1,\quad w^n:=0
\end{gather*}
with $\varepsilon':=\sgn t$ for the operator $(\hat{\mathrm D}-\lambda)^n$ and
\begin{gather*}
u=\dfrac{{\rm e}^{\theta(t,x,y)+\lambda\zeta}}{t^2+1}\sum_{j=0}^{n-1}w^j(z_1,z_2)\zeta^j,\quad
\theta(t,x,y):=-\dfrac{3t^3y^2+t(2x(t^2+1)-3ty)^2}{4(t^2+1)^3},
\\\qquad
z_1:=\dfrac y{(t^2+1)^{\frac32}},\quad
z_2:=\dfrac{(t^2+1)x-3ty}{(t^2+1)^{\frac32}},\quad
\zeta:=\arctan t,
\\
z_2w^j_1=3z_1w^j_2+w^j_{22}+(z_2^{\,2}-\lambda)w^j-(j+1)w^{j+1},\quad j=0,\dots,n-1,\quad w^n:=0
\end{gather*}
for the operator $(\hat{\mathrm P}^t+\hat{\mathrm K}-\lambda)^n$.
At the same time, even in the case $n=1$, when the corresponding invariant solutions
are in fact Lie-invariant and the corresponding reduced systems are just single equations,
the reduction procedure did not result to finding new exact solutions of the equation~\eqref{eq:RemarkableFP}
\cite[Section~5]{kova2023a}.

\begin{remark}
Solutions of the equation~\eqref{eq:RemarkableFP}
that are invariant with respect to linear differential recursion operators that are not polynomials of single Lie-symmetry operators, even the simplest among such operators, e.g., $(\mathrm P^0)^2+\mathrm P^1$,
need a separate consideration.
\end{remark}

\section{Conclusion}\label{sec:Conclusion}

The successful exhaustive classical symmetry analysis of the remarkable Fokker--Planck equation~\eqref{eq:RemarkableFP}
in~\cite{kova2023a} inspired us to study its generalized symmetries as well.
To this end, we began with computing the generalized symmetries of~\eqref{eq:RemarkableFP} up to order four
by using the excellent package {\sf Jets} by Baran and Marvan \cite{BaranMarvan} for {\sf Maple},
which is based on results of~\cite{marv2009a}.
Carefully analysing the computation results, we made two interesting observations
that allowed us to precisely conjecture the statement of Theorem~\ref{thm:GenSymsRemarkableFP}.

The first observation was that all the linear generalized symmetries of order not greater than four
are generated by the action of the Lie-symmetry operators of~\eqref{eq:RemarkableFP}
associated with the radical~$\mathfrak r$ of the algebra~$\mathfrak g^{\rm ess}$
on the elementary Lie symmetry~$u\p_u$.
In other words, $\Lambda^4=\Lambda_{\mathfrak r}^4$.

The second observation concerned the unexpected involvement of the Casimir operator
of the Levi factor~$\mathfrak f$ of~$\mathfrak g^{\rm ess}$
in the consideration of the algebra~$\Upsilon_{\mathfrak r}$.
The counterpart~$\mathrm C$ of this operator in the algebra~$\Upsilon_{\mathfrak r}$ has degree four
as a polynomial of $(\mathrm P^3,\mathrm P^2,\mathrm P^1,\mathrm P^0)$,
while it is of order three as a differential operator.
This degree--order inconsistency hinted that
straightforwardly computing the dimensions of the subspaces~$\Lambda^n$ of the algebra $\Lambda$
via evaluating the dimensions of the corresponding subspaces of the solution space
of the system of determining equations $\Delta_{kl}$, $k,l\in\mathbb N_0$, with order restrictions
is very difficult, perhaps even impossible.

Recall that the standard approach to finding the algebra of generalized symmetries
of a linear system of differential equations includes the following steps:
%\begin{list}{\arabic{enumi}.}{\usecounter{enumi}
%\labelwidth=3ex\labelsep=1ex\leftmargin=3ex
%\topsep1ex\parsep1ex\itemsep0ex\partopsep1ex}
\begin{enumerate}\itemsep=0ex
\item
For each $n\in\mathbb N_0$,
compute the dimension of the space of canonical representatives of linear generalized symmetries
of order less than or equal to $n$.

\item
If all the dimensions obtained in the previous step are finite,
then apply the Shapovalov--Shirokov theorem~\cite{shap1992a} to state that the linear generalized symmetries
exhaust all generalized symmetries up to their equivalence and linear superposition of solutions.

\item
By comparing the dimensions for each fixed order~$n$,
check whether the algebra of linear generalized symmetries
is generated by the action of known linear differential recursion operators on simple seed symmetries,
in particular, by the action of Lie-symmetry operators on the elementary Lie symmetry~$u\p_u$.
%\end{list}
\end{enumerate}
For a number of systems of differential equations,
their generalized-symmetry algebras were computed via following these steps in the presented order
\cite{kova2023b,opan2020e,shap1992a}.

In contrast, we begin by showing that the entire algebra~$\Lambda$ of linear generalized symmetries
of the equation~\eqref{eq:RemarkableFP}
coincides with the algebra $\Lambda_{\mathfrak r}$ of generalized symmetries generated by the action
of the Lie-symmetry operators~$\mathrm P^3$, $\mathrm P^2$, $\mathrm P^1$ and $\mathrm P^0$
on the vector field~$u\p_u$.
In other words, we effectively start with step~3, leaving aside the dimension counting.

From the equality $\Lambda=\Lambda_{\mathfrak r}$, we derive
\smash{$\dim\Lambda^{[n]}_{\vphantom{\mathfrak r}}=\dim\Lambda_{\mathfrak r}^{[n]}$}.
At the same time, computing the dimension \smash{$\dim\Lambda_{\mathfrak r}^{[n]}$} is a nontrivial problem,
once again due to the above inconsistency between the degree and the order of the operator~$\mathrm C$.
However, we have managed to transfer the problem to the context of ring theory and algebraic geometry,
which has allowed us to overcome this issue,
prove the inequality $\dim\Lambda_{\mathfrak r}^{[n]}<\infty$ for any $n\in\mathbb N_0$
and thus apply the Shapovalov--Shirokov theorem.
This has resulted in the proof of Theorem~\ref{thm:GenSymsRemarkableFP},
thereby completing the description of the algebra~$\Sigma$ of the equation~\eqref{eq:RemarkableFP}.

A natural question to be addressed is whether there are more examples
of differential equations for which the computation of their generalized-symmetry algebras
using the approach developed in this paper is beneficial.

We also intend to extend the study of generalized symmetries
to other (1+2)-dimensional ultraparabolic Fokker--Planck equations,
in particular to prove Conjecture~8 from~\cite{kova2024a}
on the generalized-symmetry algebra of the fine Fokker--Planck equation
$u_t+xu_y=x^2u_{xx}$.

There is an important observation that
if a homogeneous linear differential equation possesses a sufficiently large number
of linearly independent essential Lie symmetries,
then all its generalized symmetries are generated by acting
with recursion operators related to such Lie symmetries on the simplest seed Lie symmetry~$u\p_u$.
Examples of this situation include
the linear (1+1)-dimensional heat equation,
the (1+1)-dimensional Klein--Gordon equation and
the remarkable Fokker--Planck equation,
where sufficient sets of recursion operators are exhausted by selections of Lie-symmetry operators,
as well as the linear Korteweg--de Vries equation,
where one in addition needs to use the inversion of a Lie-symmetry operator
associated with the space translations.
It is an open question what are necessary and sufficient conditions
for linear systems of differential equations
whose algebras of generalized symmetries are exhausted by
those generated using Lie symmetries.
Examples of the opposite situation can be constructed from the above ones
using differential substitutions like Darboux transformations
such that the essential Lie invariance algebras of the mapped equations are trivial
while their algebras of generalized symmetries are quite large. \looseness=-1

We have developed a theoretical framework
for using linear generalized symmetries of homogeneous linear systems of differential equations
or, equivalently, their linear differential recursion operators
for constructing and generating their exact solutions.
The procedure of generalized reduction has been shown to properly work
in the case of polynomials of single Lie-symmetry operators,~$\mathrm Q$,
which can be reduced to the consideration of powers of elementary binomials, $\mathrm Q-\lambda$.
The developed techniques have been efficiently applied
to the remarkable Fokker--Planck equation~\eqref{eq:RemarkableFP},
which have essentially extended results from~\cite{kova2023a,kova2024c}.

In the context of the classical group analysis,
the linear (1+1)-dimensional heat equation~\eqref{eq:HeatEq}
and the remarkable Fokker--Planck equation~\eqref{eq:RemarkableFP}
are related to each other
since they have similar Lie- and point-symmetry properties within the classes
of parabolic linear second-order partial differential equations with two independent variables
and of ultraparabolic linear second-order partial differential equations with three independent variables,
respectively.
Surprisingly, this relation manifests on the level of generalized symmetries as well.
In particular, both the respective algebras~$\Lambda_{\rm h}$ and~$\Lambda$
of linear generalized symmetries are generated by the action of the Lie-symmetry operators
associated with the radicals of the corresponding essential Lie invariance algebras
on the elementary Lie-symmetry vector fields~$u\p_u$.
Therefore, the algebras~$\Lambda_{\rm h}$ and~$\Lambda$ are isomorphic
to the Lie algebras ${\rm W}(1,\mathbb R)^{\mbox{\tiny$(-)$}}$ and ${\rm W}(2,\mathbb R)^{\mbox{\tiny$(-)$}}$,
respectively.

The above relation can be embedded in a much wider framework.
For each $n\in\mathbb N$, consider the class~$\mathcal U_n$ of
(ultra)parabolic linear second-order partial differential equations with $1+n$ independent variables
$t$, $x_1$, \dots, $x_n$ and dependent variable~$u$,
where the corresponding (symmetric) matrices of coefficients of second-order derivatives
of the dependent variable~$u$ are of rank one,
and the number $n+1$ of independent variables is essential in the sense
that none among them plays the role of a parameter even up to their point transformations.
The equation
\[
\mathcal F_n\colon\quad u_t+\sum_{i=1}^{n-1}x_iu_{x_{i+1}}=u_{x_1x_1}
\]
belongs to the class~$\mathcal U_n$,
and the equations~$\mathcal F_1$ and~$\mathcal F_2$
coincide with the equations~\eqref{eq:HeatEq} and~\eqref{eq:RemarkableFP}, respectively.
An in-depth preliminary analysis allows us to conjecture that for each $n\in\mathbb N$,
the equation~$\mathcal F_n$ is singular within the class~$\mathcal U_n$
and has the following properties.
%\begin{list}{\arabic{enumi}.}{\usecounter{enumi}
%\labelwidth=3ex\labelsep=1ex\leftmargin=3ex
%\topsep1ex\parsep1.2ex\itemsep0.3ex\partopsep1ex}
\begin{enumerate}\itemsep=0ex
\item
The dimension of the essential Lie invariance algebra~$\mathfrak g^{\rm ess}_n$ of~$\mathcal F_n$ is equal to $2n+4$,
and this algebra is isomorphic to the algebra
${\rm sl}(2,\mathbb R)\lsemioplus_{\rho_{2n-1}\oplus\rho_0}{\rm h}(n,\mathbb R)$,
see Section~\ref{sec:RemarkableFPMIA} for the notation.
The Levi factor~$\mathfrak f_n$ and the (nil)radical~$\mathfrak r_n$ of~$\mathfrak g^{\rm ess}_n$ are isomorphic to
the algebras ${\rm sl}(2,\mathbb R)$ and ${\rm h}(n,\mathbb R)$, respectively.
\item
The dimension of~$\mathfrak g^{\rm ess}_n$ is maximal among those of the essential Lie invariance algebras
of equations from the class~$\mathcal U_n$, and each equation whose essential Lie invariance algebra
is of this maximal dimension is reduced to~$\mathcal F_n$ by a point transformation
in the space $\mathbb R^{1+n}_{t,x_1,\dots,x_n}\times\mathbb R_u$.
\item
The essential point-symmetry group~$G^{\rm ess}_n$ of the equation~$\mathcal F_n$ is isomorphic to the Lie group
$\big({\rm SL}(2,\mathbb R)\ltimes_{\varrho_{2n-1}\oplus\varrho_0}{\rm H}(n,\mathbb R)\big)\times\mathbb Z_2$,
where ${\rm H}(n,\mathbb R)$ denotes the rank-$n$ Heisenberg group
and~$\varrho_m$ is the irreducible representation of ${\rm SL}(2,\mathbb R)$ in $\mathbb R^{m+1}$.
\item
A complete list of discrete point symmetry transformations of the equation~$\mathcal F_n$
that are independent up to combining with each other
and with continuous point symmetry transformations of this equation
is exhausted by the single involution~$\mathscr I$ alternating the sign of~$u$,
$\mathscr I\colon(t,x_1,\dots,x_n,u)\mapsto(t,x_1,\dots,x_n,-u).$
Thus, the quotient group of the complete point-symmetry pseudogroup~$G_n$ of~$\mathcal F_n$
with respect to its identity component is isomorphic to~$\mathbb Z_2$.
\item
The algebra of canonical representatives of generalized symmetries of~$\mathcal F_n$ is
$\Sigma_n=\Lambda_n\lsemioplus\Sigma^{-\infty}_n$.
Here
$\Lambda_n$ is the subalgebra of linear generalized symmetries of~$\mathcal F_n$,
which is generated by acting with the Lie-symmetry operators associated with
the canonical basis of the complement of the center $\langle u\p_u\rangle$
in the (nil)radical~$\mathfrak r_n$ of~$\mathfrak g^{\rm ess}_n$
on the elementary seed symmetry vector field $u\p_u$,
and
$\Sigma^{-\infty}_n$ is the ideal associated with linear superposition of solutions of~$\mathcal F_n$.
\item
The algebra~$\Lambda_n$ is isomorphic to the Lie algebra ${\rm W}(n,\mathbb R)^{\mbox{\tiny$(-)$}}$
associated with the $n$th Weyl algebra ${\rm W}(n,\mathbb R)$,
$\Lambda_n\simeq {\rm W}(n,\mathbb R)^{\mbox{\tiny$(-)$}}$.
Hence, the algebra $\Lambda_n$ is $\mathbb Z$-graded.
\item
The algebra~$\Lambda_n$ is two-generated as a Lie algebra,
i.e., there is a pair of its elements such that $\Lambda_n$ coincides with its subalgebra
containing all successive commutators (aka nonassociative monomials) of these two elements.
\end{enumerate}
%\end{list}
The complete and rigorous proofs of the listed properties
constitute the subject of a substantial research program
whose realization will result in a deeper understanding of symmetry properties
of linear second-order partial differential equations.
It will essentially extend the results of \cite{kova2023b,lie1881a,olve1993A}
on the linear (1+1)-dimensional heat equation~$\mathcal F_1$
and of~\cite{kova2023a} and this paper
on the remarkable Fokker--Planck equation~$\mathcal F_2$
to~$\mathcal F_n$ with an arbitrary $n\in\mathbb N$.

Establishing the isomorphism between the algebras~$\Upsilon_{\mathfrak r}$ and ${\rm W}(2,\mathbb R)$
(resp.\ $\Lambda$ and ${\rm W}(2,\mathbb R)^{\mbox{\tiny$(-)$}}$)
allows us to transfer the results naturally obtained for one of them to the other,
and the mapped results can be not as apparent as their counterparts.
See Remarks~\ref{rem:LamdbaTwo-Generation} and~\ref{rem:W2R-Grading},
where the transfers go from the abstract algebras to their realizations.
Unexpected examples of opposite transfers are given by filtrations of the above associative algebras.
Both the filtration~$F_2$ of the algebra~$\Upsilon_{\mathfrak r}$
with respect to the degree of its elements as (noncommutative) polynomials
in $\{\mathrm P^0,\mathrm P^1,\mathrm P^2,\mathrm P^3\}$, which is given in Section~\ref{sec:RemarkableFPMIA},
and its counterpart for the algebra~${\rm W}(2,\mathbb R)$ are natural
in the context of general noncommutative polynomial algebras.
At the same time, we see no way to naturally interpret, for the second Weyl algebra~${\rm W}(2,\mathbb R)$,
the image of the other filtration of~$\Upsilon_{\mathfrak r}$ presented therein
and associated with the order of elements of~$\Upsilon_{\mathfrak r}$ as differential operators.
There are other similar natural filtrations of~$\Upsilon_{\mathfrak r}$
that are related to various interpretations of the order of differential operators.
In particular, we can assign to each element of~$\Upsilon_{\mathfrak r}$
the order of its counterpart obtained by excluding multiple derivatives with respect to~$x$
according to the Kovalevskaya form $u_{xx}=u_t+xu_y$ of the equation~\eqref{eq:RemarkableFP}.
The discussion in this paragraph is of significant interest since, as far as we know,
the general problems of classifying filtrations and gradings
of the Weyl algebras and of the associated Lie algebras have not been solved.

\section*{Acknowledgments}
%\bigskip\par\noindent{\bf Acknowledgments.}
The authors are grateful to Alexander Bihlo, Yuri Bahturin, Mikhail Kochetov, Vyacheslav Boyko and Galyna Popovych
for valuable discussions.
They are also deeply thankful to anonymous reviewers for excellent remarks and suggestions,
which helped to essentially improve the paper.
This research was undertaken thanks to funding from the Canada Research Chairs program,
the InnovateNL LeverageR{\&}D program and the NSERC Discovery Grant program.
It was also supported in part by the Ministry of Education, Youth and Sports of the Czech Republic (M\v SMT \v CR)
under RVO funding for I\v C47813059.
ROP expresses his gratitude for the hospitality shown by the University of Vienna during his long stay at the university.
The authors express their deepest thanks to the Armed Forces of Ukraine and the civil Ukrainian people
for their bravery and courage in defense of peace and freedom in Europe and in the entire world from russism.
\looseness=-1

%\newpage

%\noindent
%{\bf Competing interests:} The authors declare none.

\end{document}